\documentclass{cccg20}
\usepackage{color,ifpdf,latexsym,graphicx,url,amssymb,cite,microtype,subcaption}
\DeclareUrlCommand\path{\urlstyle{same}}
\title{New Results in Sona Drawing: Hardness and TSP Separation}
\author{%
  Man-Kwun Chiu%
    \thanks{Institut f\"ur Informatik, 
    	Freie Universit\"at Berlin, Germany,
    \protect\url{chiumk@inf.fu-berlin.de}.
    This work was supported in part by ERC StG 757609.}
\and
  Erik D. Demaine%
    \thanks{CSAIL, Massachusetts Institute of Technology, USA, \protect\url{{edemaine,diomidov,achester}@mit.edu}}
\and
  Yevhenii Diomidov%
    \footnotemark[2]
\and
  David Eppstein%
    \thanks{Computer Science Department,
        University of California, Irvine, \protect\url{eppstein@uci.edu}. This work was supported in part by the US National Science Foundation under grant CCF-1616248.}
\and
  Robert A. Hearn%
    \thanks{\protect\url{bob@hearn.to}}
\and
  Adam Hesterberg%
    \footnotemark[2]
\and
  Matias Korman%
    \thanks{Tufts University, USA, \protect\url{matias.korman@tufts.edu}}
\and
  Irene Parada%
    \thanks{TU Eindhoven, The Netherlands, \protect\url{i.m.de.parada.munoz@tue.nl}}
\and
  Mikhail Rudoy%
    \thanks{CSAIL, MIT, USA. Now at Google Inc.}
}

\newif\ifabstract
\abstracttrue
\newif\iffull
\ifabstract \fullfalse \else \fulltrue \fi

\newcounter{section-preserve}
\newcounter{theorem-preserve}
\newcommand{\blank}[1]{}
\newtoks\magicAppendix
\magicAppendix={}
\newtoks\magictoks
\newif\iflater
\laterfalse
\long\def\later#1{\magictoks={#1}%
  \edef\magictodo{\noexpand\magicAppendix={\the\magicAppendix \par
    \the\magictoks%
  }}
  \magictodo}
\long\def\both#1{\magictoks={#1}%
  \edef\magictodo{\noexpand\magicAppendix={\the\magicAppendix \par
    \noexpand\setcounter{theorem-preserve}{\noexpand\arabic{theorem}}%
    \noexpand\setcounter{theorem}{\arabic{theorem}}%
    \noexpand\setcounter{section-preserve}{\noexpand\arabic{section}}%
    \noexpand\setcounter{section}{\arabic{section}}%
    \noexpand\let\noexpand\oldsection=\noexpand\thesection
    \noexpand\def\noexpand\thesection{\thesection}
    \noexpand\let\noexpand\oldlabel=\noexpand\label
    \noexpand\let\noexpand\label=\noexpand\blank
    \the\magictoks%
    \noexpand\setcounter{theorem}{\noexpand\arabic{theorem-preserve}}%
    \noexpand\setcounter{section}{\noexpand\arabic{section-preserve}}%
    \noexpand\let\noexpand\thesection=\noexpand\oldsection
    \noexpand\let\noexpand\label=\noexpand\oldlabel
  }}
  \magictodo
  \the\magictoks}
\def\magicappendix{\latertrue \the\magicAppendix}

\usepackage{hyperref}
\hypersetup{breaklinks,bookmarks,bookmarksnumbered,bookmarksopen,bookmarksopenlevel=2}
{\makeatletter \hypersetup{pdftitle={\@title}}}
{\makeatletter
 \gdef\xxxmark{%
   \expandafter\ifx\csname @mpargs\endcsname\relax % in minipage?
     \expandafter\ifx\csname @captype\endcsname\relax % in figure/caption?
       \marginpar{xxx}% not in a caption or minipage, can use marginpar
     \else
       xxx % notice trailing space
     \fi
   \else
     xxx % notice trailing space
   \fi}
 \gdef\xxx{\@ifnextchar[\xxx@lab\xxx@nolab}
 \long\gdef\xxx@lab[#1]#2{\textbf{[\xxxmark #2 ---{\sc #1}]}}
 \long\gdef\xxx@nolab#1{\textbf{[\xxxmark #1]}}
}

{\makeatletter \gdef\fps@figure{!htbp}}

\let\realbfseries=\bfseries
\def\bfseries{\realbfseries\boldmath}

\def\compactify{\itemsep=0pt \topsep=0pt \partopsep=0pt \parsep=0pt}
\let\latexusecounter=\usecounter
\newenvironment{enumerate*}
  {\def\usecounter{\compactify\latexusecounter}
   \begin{enumerate}}
  {\end{enumerate}\let\usecounter=\latexusecounter}

\let\epsilon=\varepsilon
\def\defn#1{\textbf{\textit{\boldmath #1}}}
\let\emph=\defn
\def\mid{\mathrel{:}}

\usepackage{amsmath}
\def\TSP{\mathrm{TSP}}

\newcommand{\irene}[1]{{\color{red}I.P: #1}}

\begin{document}
\maketitle

\begin{center}\sl
  In memoriam Godfried Toussaint (1944--2019)
\end{center}

\begin{abstract}
Given a set of point sites, a sona drawing is a single closed curve, disjoint from the sites and intersecting itself only in simple crossings, so that each bounded region of its complement contains exactly one of the sites.
We prove that it is NP-hard to find a minimum-length sona drawing for $n$ given points, and that such a curve can be longer than the TSP tour of the same points by a factor $> 1.5487875$.  When restricted to tours that lie on the edges of a square grid, with points in the grid cells, we prove that it is NP-hard even to decide whether such a tour exists. These results answer questions posed at CCCG 2006.
\end{abstract}

\section{Introduction}

In April 2005, Godfried Toussaint visited the second author at MIT,
%bringing with him the idea of analyzing
where he proposed a computational geometric analysis of the ``sona''
sand drawings of the Tshokwe people in the West Central Bantu area of Africa.
Godfried encountered sona drawings, in particular the ethnomathematical
work of Ascher \cite{ascher-02} and Gerdes \cite{gerdes-99},
during his research into African rhythms.
Together with his then-student Perouz Taslakian,
we came up with a formal model of \defn{sona drawing}
of a set $P$ of point \defn{sites} ---
a closed curve drawn in the plane such that
\begin{enumerate}\itemsep=0pt
\item wherever the curve touches itself, it crosses itself;
\item each crossing involves only two arcs of the curve;
\item exactly one site is in each bounded face formed by the curve; and
\item no sites lie on the curve or within its outside face.
\end{enumerate}

Our first paper on sona drawings appeared at BRIDGES 2006 \cite{DDTT07},
detailing the related cultural practices, proving and computing
combinatorial results and drawings, and posing several open problems.
In early 2006, we brought these open problems to Godfried's
Bellairs Winter Workshop on Computational Geometry, where a much larger group
tackled sona drawings, resulting in a CCCG 2006 paper later the same year
\cite{DDD06}.  Next we highlight some of the key prior results and
open problems as they relate to the results of this paper.

\paragraph{Sona vs.\ TSP.}
Every TSP tour can be easily converted into a sona drawing of roughly the same
length: instead of visiting a site, loop around it, except for one site that
we place slightly interior to the tour \cite[Lemma~11]{DDTT07}.
Conversely, every sona drawing can be converted into a TSP tour of length
at most a factor ${\pi+2 \over \pi} \approx 1.63661977$ larger
\cite[Theorem~12]{DDD06}, settling \cite[Open Problem~6]{DDTT07}.
Is this constant tight?
The best previous lower bound was a four-site example proving a TSP/sona
separation factor of ${2 \over 3} + {2 \sqrt 3 \over 9} \approx 1.05156685$
\cite[Lemma~12]{DDTT07}.
In \autoref{sec:separation},
we construct a recursive family of examples proving a
much larger TSP/sona separation factor of
$\frac{14 + 8 \sqrt{2} + \pi(\sqrt{2}+1)}{8 + 4\sqrt{2} + \pi(\sqrt{2}+1)}
\approx 1.54878753$.
We also study $L_1$ and $L_\infty$ metrics, where we prove that the worst-case
TSP/sona separation factor is exactly $1.5$.

\paragraph{Length minimization.}
The relation to TSP implies a constant-factor approximation algorithm for
finding the minimum-length sona drawing on a given set of sites.
But is this problem NP-hard?  In \autoref{sec:min-length},
we prove NP-hardness for $L_1$, $L_2$, and $L_\infty$ metrics,
settling \cite[Open Problem~5]{DDTT07} and
\cite[Open Problem~4]{DDD06}.

\paragraph{Grid drawings.}
The last variant we consider is when the sona drawing is restricted to
lie along the edges of a unit-square grid, while sites are at the centers of
cells of the grid.
Not all point sets admit a grid sona drawing;
however, if we scale the sites' coordinates by a factor of $3$,
then they always do \cite[Proposition~10]{DDD06}.
A natural remaining question \cite[Open Problem~3]{DDD06}
is which point sets admit grid sona drawings.
In \autoref{sec:grid}, we prove that this question is in fact NP-hard.

\section{Separation from TSP Tour}
\label{sec:separation}
\later{\section{\boldmath Separation from TSP Tour under the $L_2$ Metric} \label{app:separationL2}}

We first show an example which gives a large TSP/sona
separation factor under the $L_2$ metric in the plane.

\both{
\begin{theorem} \label{thm:ratio}
There exists a set of sites for which the length of the  minimum-length TSP tour is $\frac{14 + 8 \sqrt{2} + \pi(\sqrt{2}+1)}{8 + 4\sqrt{2} + \pi(\sqrt{2}+1)} \approx 1.54878753$ times the length of the minimum-length sona drawing.
\end{theorem}
}
\begin{figure*}
  \centering
  \subcaptionbox{First step of our construction for $\epsilon^{-1}=2$: points of $A_0$ lie in the intersection of solid lines (packed segments). In the construction, $P_1$ contains thirteen sites (shown as black dots).
  \label{fig:minsona_a}}{\includegraphics[width=0.47\linewidth]{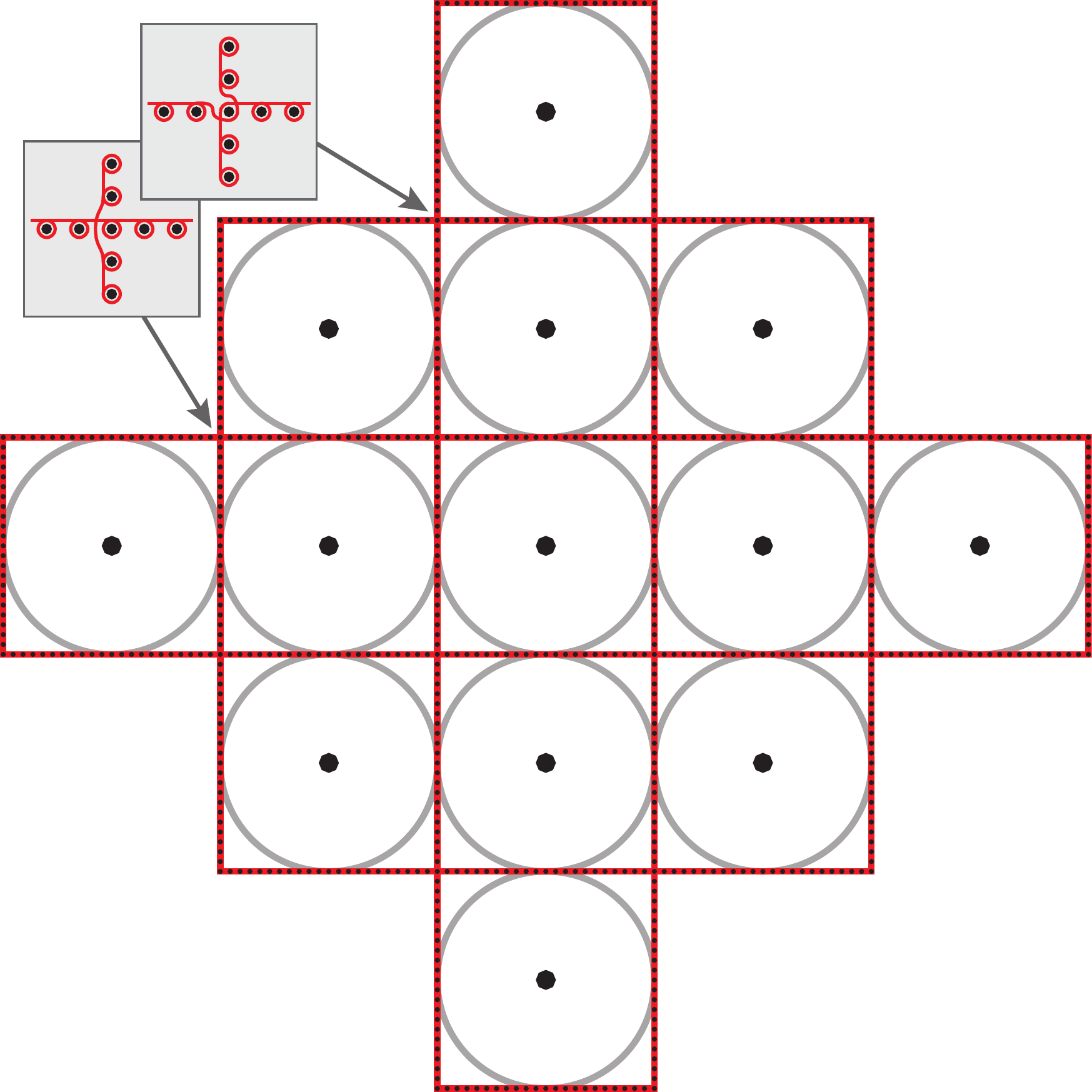}}\hfill
  \subcaptionbox{Final construction for $\epsilon^{-1}=2$. Sets $P_1$, $P_2$, and $P_3$ are shown as black dots of varying sizes. Additional sites of $P^{(2)}$ pack lines and rounded squares nearby the auxiliary points of $A_0$, $A_1$, and $A_2$.\label{fig:minsona_b}}{\includegraphics[width=0.47\linewidth]{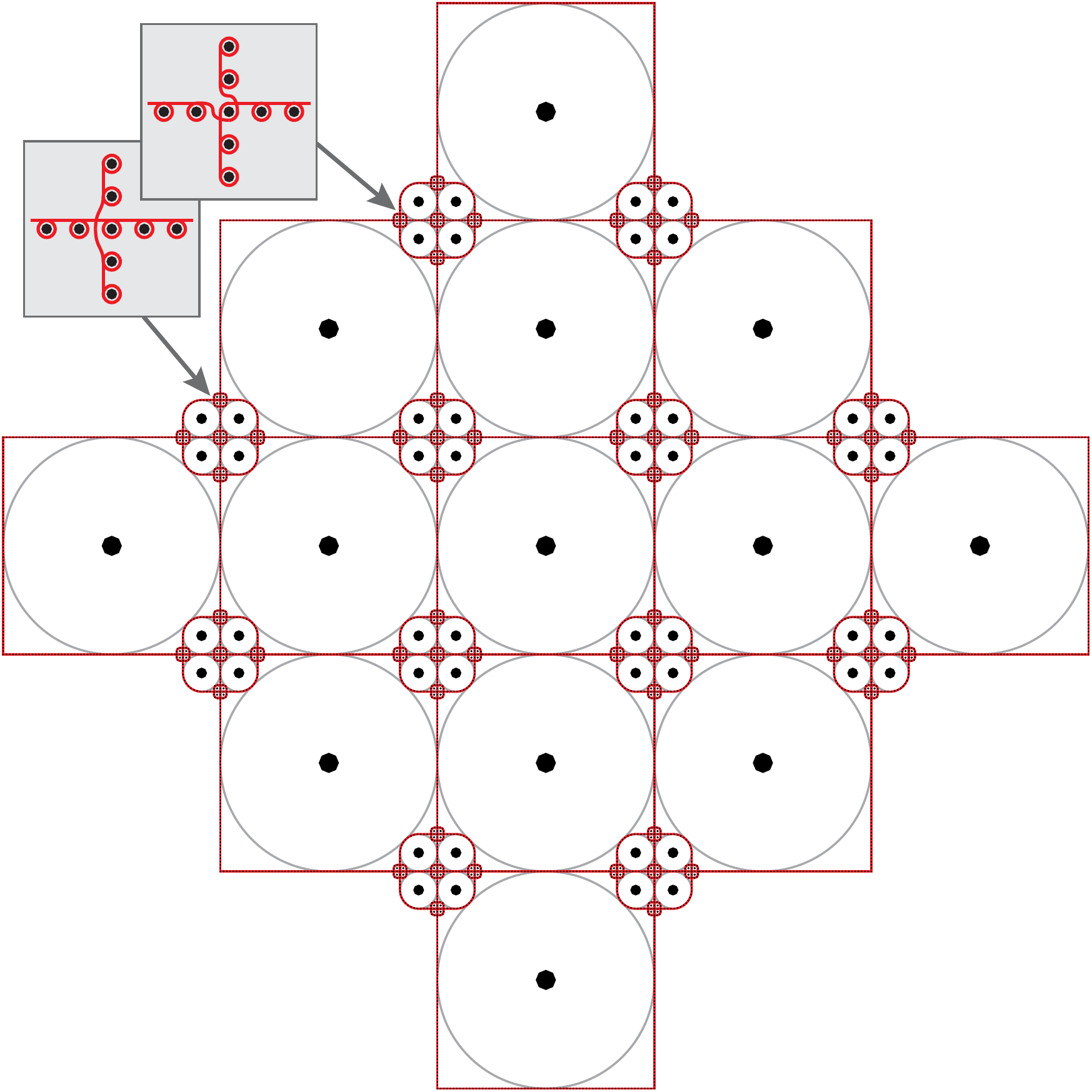}}
  \caption{Recursive construction of sites requiring
    %$\frac{14 + 8 \sqrt{2} + \pi(\sqrt{2}+1)}{8 + 4\sqrt{2} + \pi(\sqrt{2}+1)}
    $\approx 1.54878753$ factor shorter sona tour (drawn in red)
    compared to TSP tour (red plus doubled radius of each grey circle).
    All red lines have black sites sprinkled densely along them.}
  \label{fig:minsona}
\end{figure*}
The full proof can be found in Appendix~\ref{app:separationL2}. 
\medskip

\later{\begin{proof}}
\noindent\textbf{Sketch of Proof.}%
\both{%
We construct a problem instance whose minimum-length TSP tour is longer than its minimum-length sona drawing by a factor within $\epsilon$ of $\frac{14 + 8 \sqrt{2} + \pi(\sqrt{2}+1)}{8 + 4\sqrt{2} + \pi(\sqrt{2}+1)}$. Our construction is illustrated in Figure~\ref{fig:minsona} and follows a fractal approach with $\epsilon^{-1}$ levels (for simplicity, we assume that $\epsilon^{-1}$ is an integer).

\begin{enumerate}
\item We start by defining a few auxiliary points:
\begin{enumerate}
\item The initial set of auxiliary points $A_0$ is the intersection between a slightly shifted integer lattice and the $L_1$ ball $B(0,\epsilon^{-1})$ of radius $\epsilon^{-1}$ centered at the origin. That is, $A_0=\left\{(\frac{2i+1}{2},\frac{2j+1}{2}) \mid i,j \in \mathbb{Z}\right\} \cap \left\{(x,y) \mid |x|+|y| \leq \epsilon^{-1}\right\}$. In Figure~\ref{fig:minsona_a}, the auxiliary points are exactly the intersections of solid red lines. 
\item Then, in Step $i$ (starting with $i=1$), for each auxiliary point $p\in A_{i-1}$, we add five points to $A_i$: $p$ itself, and four new points at distance $\left(\frac{1}{1+\sqrt{2}}\right)^{2i}$ from $p$ in each of the four cardinal directions. Let $A= A_{\epsilon^{-1}}$. By construction, we have $|A_i|=5^i|A_0|$ and $|A_0|=2\epsilon^{-2} + O(\epsilon^{-1})$. We note that set $A$ contains auxiliary points (not sites). These points will not be part of the instance.
\end{enumerate}

\item We now use the auxiliary points to create some sites (the isolated black points of Figure~\ref{fig:minsona}):
\begin{enumerate}
\item For any $i\geq 1$ we define set $P_i$ of sites as follows: for each auxiliary point $q\in A_{i-1}$ we add the four sites 
 whose $x$ and $y$ coordinates each differ from $q$ by $\frac12 \left(\frac{1}{1+\sqrt{2}}\right)^{2i}$. 
 We note that all the added sites are distinct sites. 
\item We define $P_0$ as the set of integer lattice points in $B(0,\epsilon^{-1})$. Equivalently, for each auxiliary point $q\in A_{0}$ we add the sites whose $x$ and $y$ coordinates each differ from $q$ by $\frac12$, but we do not add sites that lie outside $B(0,\epsilon^{-1})$, which affects $O(\epsilon^{-1})$ sites (out of $\Omega(\epsilon^{-2})$ sites of $P_0$). 
 
%In terms of cardinality, observe that for $i \ge 1$, we add four sites per auxiliary point $q \in A_{i-1}$. The only exception is when $i = 0$:  
%We note that, 
In this case, the sites created by different auxiliary points may lie in the same spot. In total, $P_0$ contains only one site per auxiliary point of $A_0$ (except for $O(\epsilon^{-1})$ auxiliary points near the boundary). Thus, $|P_i|=4\cdot|A_{i-1}|=4\cdot 5^{i-1}\cdot |A_0|$ (for $i\geq 1$) and $|P_0|=|A_0|+O(\epsilon^{-1})=(2\epsilon^{-2} + O(\epsilon^{-1}))$. 
\end{enumerate}

Let $P^{(1)}=P_0 \cup P_1 \cup \cdots \cup P_{\epsilon^{-1}}$. 

\item\label{pack} Next, we place additional sona sites packing line segments and/or curves. Whenever we pack any curve, we place sites spaced at a distance $\delta$ small enough that the length of the shortest path that passes within $\delta$ of all of them is within a factor $(1-\epsilon)$ of the length of the curve. 
\begin{enumerate}
\item Solid lines as drawn in red in Figure~\ref{fig:minsona_a}: for each $x \in \{i+\frac{1}{2} \mid  -(\epsilon^{-1}+1) \leq i \leq \epsilon^{-1} \mbox{ and } i \in \mathbb{Z}\}$,
%\cap \{(x,y) \mid |x|+|y|< \epsilon^{-1}\}$
we pack the vertical line segment with endpoints $(x, \epsilon^{-1}+1 - |x|)$ and $(x, |x| - \epsilon^{-1}-1)$, and analogously with $y$ for the horizontal line segments.
%with $x$ and $y$ switched.
%are related to $\in A_0$: for any two {\em adjacent} points $p,q\in A_0$ (that is, they share the $x$ coordinate and $y$ coordinates differ by $1/2$ or {\em vice versa}), we pack the line segment with endpoints $p$ and $q$. Note that every point of $A_0$ lies at the intersection of packed segments (see Figure~\ref{fig:minsona}, left). \textcolor{red}{Mati sats: Error, this is not totally accurate}
\item In Step $i$ of the above recursive definition (starting with $i=1$), when we create four new auxiliary points of $A_i$ from a point $p\in A_{i-1}$, we also pack the boundary of the region within Euclidean distance $\frac12 \left(\frac{1}{1+\sqrt{2}}\right)^{2i}$ of the square whose vertices are the four auxiliary points of~$A_i$. Note that this boundary region forms a square with rounded corners as in Figure~\ref{fig:minsona_b}. With these extra points we preserve the invariant that the auxiliary points are exactly the intersections of packed curves.
\end{enumerate}
\end{enumerate}

Let $P^{(2)}$ be the set of sites created in Step \ref{pack} in our construction, and $P=P^{(1)} \cup P^{(2)}$. This is a complete description of the construction.}

{
In the full proof we show that the length of the packed curves is a $(1+\epsilon)$-approximation of the total length of the minimum-length sona drawing of $P$. 
Careful calculations then yield that the length of the minimum-length sona drawing is $(2\epsilon^{-2} + O(\epsilon^{-1})) \left(2 + \frac{4 + \pi}{2\sqrt{2}-2}\right).$
We then argue that the minimum-length TSP has an additional length of $(2\epsilon^{-2} + O(\epsilon^{-1}))(2\sqrt{2}+3)$. 
Thus, the TSP/sona separation factor for the construction is, ignoring lower-order terms, $\frac{2 + \frac{4 + \pi}{2\sqrt{2}-2} + 2 \sqrt{2} + 3}{2 + \frac{4 + \pi}{2\sqrt{2}-2}} = \frac{14 + \pi(\sqrt{2}+1) + 8 \sqrt{2}}{8 + 4\sqrt{2} + \pi(\sqrt{2}+1)} \approx 1.54878753$.
}

\later{We now find the minimum-length sona drawing of $P$. 
Each pair of consecutive points in a packed curve must be in a separate sona region, so any sona drawing must pass between them; in particular, any sona drawing must pass within $\delta$ of each of them, and so the length of any valid sona drawing is at least $1-\epsilon$ times the length of the packed curves. Also, there's a valid sona drawing that's at most $1+\epsilon$ times the length of the packed curves: follow all the packed curves exactly, adding small loops around the sona sites of the packed curves as necessary (loops small enough to lengthen the curve by a factor of at most $1 + \epsilon$). The graph of packed curves is Eulerian (because it's defined as a union of boundaries of regions, which are cycles), so the TSP tour can follow an Eulerian circuit through it. At an intersection of packed curves, we have two options for the sona drawing (as shown in the inset images of Figure~\ref{fig:minsona}). We can have one sona path cross over the other in the Eulerian circuit (by including every site of the packed curve in a small loop). Alternatively, we can have one sona path cross over the other in two places $p$ and $q$ at the intersection, and leaving one sona site of the packed curve out of a small loop to be the sona site of the extra region between $p$ and $q$. In either case, we conclude that there is a valid sona path that follows an Eulerian circuit of the packed curves within $(1+\epsilon)$. Note that, although our description focused in the sites of $P^{(2)}$, this is a valid sona tour for $P$ since the sites of $P^{(1)}$ lie in different faces.

The total length of the packed curves is hence a $(1+\epsilon)$-approximation of the total length of the minimum-length sona drawing.  

The total length of the packed segments of $P^{(2)}$ (the square lattice) is $4\epsilon^{-2} + O(\epsilon^{-1})$, since the area of the region $|x| + |y| < \epsilon^{-1}$ and the number of lattice points in it are each $2\epsilon^{-2} + O(\epsilon^{-1})$.

Now we bound the length of the packed curves (rounded squares). In Step $i$ of the construction (starting with $i=1$), we added a packed curve that is the boundary of the region within Euclidean distance $\frac12 \left(\frac{1}{1+\sqrt{2}}\right)^{2i}$ of a square of side length $\left(\frac{1}{1+\sqrt{2}}\right)^{2i}$ (with a total length of $(4 + \pi) \left(\frac{1}{1+\sqrt{2}}\right)^{2i}$). 

Recall that we added one such curve for each of the points of $A_{i-1}$ and that $|A_i|=5^{i-1}(2\epsilon^{-2} + O(\epsilon^{-1}))$. Thus, the total length of the sona drawings introduced at Step $i$ is $(4 + \pi) \left(\frac{1}{1+\sqrt{2}}\right)^{2i} 5^{i-1}(2\epsilon^{-2} + O(\epsilon^{-1}))$. 

For $\epsilon$ small, this series is well-approximated by an infinite geometric series with sum 
$(2\epsilon^{-2} + O(\epsilon^{-1})) \left(\frac{4 + \pi}{(1 + \sqrt{2})^2 \cdot\left(1-\frac{5}{(1 + \sqrt{2})^2}\right)}\right) = (2\epsilon^{-2} + O(\epsilon^{-1})) \left(\frac{4 + \pi}{2\sqrt{2}-2}\right)$, and adding in the length of the packed segments of $P^{(2)}$ (the square lattice) gives $(2\epsilon^{-2} + O(\epsilon^{-1})) \left(2 + \frac{4 + \pi}{2\sqrt{2}-2}\right).$

%\textcolor{red}{Mati. I do not see where the $2+$ comes from. My calculation is as follows:
%\begin{eqnarray}
%&=&\sum_{i \geq 1} (2\epsilon^{-2} + O(\epsilon^{-1})) \frac{4+\pi}{(1+\sqrt{2})^2} (\frac{5}{(1+\sqrt{2})^2})^{i-1} \\
%&=&(2\epsilon^{-2} + O(\epsilon^{-1})) \frac{4+\pi}{3+2\sqrt{2}} (\frac{1}{1-\frac{5}{(1+\sqrt{2})^2}})\\
%&=&(2\epsilon^{-2} + O(\epsilon^{-1})) \frac{4+\pi}{2\sqrt{2}-2} \\
%\end{eqnarray}
%}
%The boundary of that region has four straight segments of length $\left(\frac{1}{1+\sqrt{2}}\right)^{2i}$ and four circular arcs making the full circumference of a circle of radius $\frac12 \left(\frac{1}{1+\sqrt{2}}\right)^{2i}$, for a total length of $(4 + \pi) \left(\frac{1}{1+\sqrt{2}}\right)^{2i}$ at each lattice point, or $(4 + \pi) \left(\frac{1}{1+\sqrt{2}}\right)^{2i} (2\epsilon^{-2} + O(\epsilon^{-1}))$ total at Step $i$. So the total length of the sona drawing is $(2\epsilon^{-2} + O(\epsilon^{-1})) \left(2 + (4 + \pi)\sum_{i=1}^{\epsilon^{-1}} 5^{i-1} \left(\frac{1}{1+\sqrt{2}}\right)^{2i} \right)$. For $\epsilon$ small, this series is well-approximated by an infinite geometric series with sum $(2\epsilon^{-2} + O(\epsilon^{-1})) \left(2 + \frac{4 + \pi}{((1 + \sqrt{2})^2)(1-\frac{5}{(1 + \sqrt{2})^2})}\right) = (2\epsilon^{-2} + O(\epsilon^{-1})) \left(2 + \frac{4 + \pi}{2\sqrt{2}-2}\right)$.

We have approximated the minimum length of a valid sona drawing; now we approximate the minimum length of a TSP tour. 

Any TSP tour must also come within $\delta$ of every point on every packed curve, which requires a length at least $(2\epsilon^{-2} + O(\epsilon^{-1})) \left(2 + \frac{4 + \pi}{2\sqrt{2}-2}\right)$ as above.
%\textcolor{red}{Mati: we hsould update if my numbers above are right}. 
Also, the TSP tour must visit each site of $P^{(1)}$. We observe some properties of this set:

\begin{itemize}
\item Set $P^{(1)}$ is defined so that sites are far from each other. Specifically, the Euclidean ball centered at any site $p\in P_i$ of radius $r_i=\frac12 \left(\frac{1}{1+\sqrt{2}}\right)^{2i}$ does not contain other sona sites. This means that we must include at least $2r_i$ in the length of the TSP tour for each point in $P_i$, for the part of the tour that passes from the boundary of this ball to $P_i$ and then back to the boundary.

\item There are $4\cdot 5^{i-1}(2\epsilon^{-2} + O(\epsilon^{-1}))$ sites in $P_i$ (for $i\geq 1$) and $(2\epsilon^{-2} + O(\epsilon^{-1}))$ sites in $P_0$.
\end{itemize}

When $\epsilon$ tends to zero, the additional length needed in the TSP tour is 
\begin{eqnarray*}
&&(2\epsilon^{-2} + O(\epsilon^{-1}))(1+\sum_{i\geq 1} 2r_i \cdot 4\cdot 5^{i-1})\\
&=&(2\epsilon^{-2} + O(\epsilon^{-1})) (1+\frac45 \sum_{i \geq 1} \left(\frac{5}{3+2\sqrt{2}}\right)^{i}) \\
&=&(2\epsilon^{-2} + O(\epsilon^{-1})) (1+\frac{4}{3+2\sqrt{2}} \cdot \frac{1}{1- \frac{5}{3+2\sqrt{2}}})\\
&=&(2\epsilon^{-2} + O(\epsilon^{-1}))(1+\frac{4}{2\sqrt{2}-2}) \\
&=&(2\epsilon^{-2} + O(\epsilon^{-1}))(2\sqrt{2}+3). 
\end{eqnarray*}

%

\iffalse
 In Step $0$ we add, for each lattice point, an isolated point at distance $\frac12$ from every point on a packed curve:
 \irene{I find the phrasing of the next two points confusing.}
\begin{enumerate}
\item It's at distance $\frac12$ from the integer lattice grid. 
\item The square with rounded corners is at most $\frac12\left(\frac{1}{1+\sqrt{2}}\right)^{2}(1+\sqrt{2})$ from a lattice point: $\frac12\left(\frac{1}{1+\sqrt{2}}\right)^{2}(\sqrt{2})$ from the lattice point to the center of a circular arc, and $\frac12\left(\frac{1}{1+\sqrt{2}}\right)^{2}(1)$ from there to the edge of a circular arc. The distance from the isolated point to the lattice point is $\frac12\sqrt{2}$, so the distance to the square with rounded corners is $\frac12 \sqrt{2} - \frac12\left(\frac{1}{1+\sqrt{2}}\right)^{2}(1+\sqrt{2}) = \frac12$.
\end{enumerate}
So visiting that point requires two edges of length $\frac12$. Even without any vertices on packed curves that are visited immediately before or after, the packed curves are densely packed enough that the TSP tour needs to traverse almost their full length.

Similarly, starting with $i = 0$, at Step $i$ we add $4 \cdot 5^i$ isolated vertices per lattice point, each of which is at distance at least $\frac12\left(\frac{1}{1+\sqrt{2}}\right)^{2i+2}$ from any packed edge, requiring the TSP tour to take an extra length of $4 \cdot 5^i\left(\frac{1}{1+\sqrt{2}}\right)^{2i+2}$. Summing that over $i \ge 0$ gives $\frac{4}{(1+\sqrt{2})^2(1-\frac{5}{(1+\sqrt{2})^2})} = 2 \sqrt{2} + 2.$
\fi

So, the total length of the TSP tour is at least $(2\epsilon^{-2} + O(\epsilon^{-1})) \left(2 + \frac{4 + \pi}{2\sqrt{2}-2} + 2 \sqrt{2} + 3\right)$. Hence the ratio of the length of the TSP tour to the length of the sona drawing is, ignoring lower-order terms, $\frac{2 + \frac{4 + \pi}{2\sqrt{2}-2} + 2 \sqrt{2} + 3}{2 + \frac{4 + \pi}{2\sqrt{2}-2}} = \frac{14 + \pi(\sqrt{2}+1) + 8 \sqrt{2}}{8 + 4\sqrt{2} + \pi(\sqrt{2}+1)} \approx 1.54878753$.
\end{proof}}
\medskip

We have presented a construction for the $L_2$ metric in the plane showing that the ratio between the lengths of the minimum-length TSP and the minimum-length sona drawing can be strictly greater than 1.5. 
Our next result shows that this cannot be the case for the $L_1$ and~$L_\infty$ metrics in the plane.
%Apart from the $L_2$ metric in the plane, 
%we have some other results for the $L_1$ and $L_\infty$ metrics. In particular, we show that a ratio $1.5$ is tight for the $L_1$ and $L_\infty$ metrics in the plane. It means that the length of the minimum-length TSP tour is always at most $1.5$ times the length of the minimum-length sona drawing, but sometimes can be very close to $1.5$ times the length of the minimum-length sona drawing.

\begin{theorem}
	\label{thm:L1Linf}
	For the Manhattan ($L_1$) and the Chebyshev ($L_{\infty}$) metrics,
	the minimum-length TSP tour for a set of sites $P$ has length
	at most $1.5$ times that of the minimum-length sona drawing for~$P$. 
	Moreover, this bound is tight for both metrics. 
\end{theorem}

\begin{proof}
	The proof of the upper bound on the length of the  minimum-length TSP tour follows the lines of the (unpublished) 
	proof of~\cite[Theorem 12]{DDD06}. 
	Let $P = \{p_1,\ldots, p_n\}$ be a set of $n$ sites, 
	$S(P)$ the minimum-length sona drawing for $P$,  
	and $\TSP(P)$ the minimum-length TSP tour for $P$. 
	The sona drawing $S(P)$ must have $n$ bounded faces, each containing a site of $P$. 
	Let~$f_i$ be the face of $S(P)$ containing the site $p_i$. 
	In this proof, for an edge-weighted graph $H$, $|H|$ denotes the sum of the weights/lengths of all the edges of $H$. 
	In particular, $|S(P)|$ denotes the length of the sona drawing~$S(P)$. 
	
	For each site $p_i$, let $c(p_i)$ be the closest point in $S(P)$ to $p_i$ 
	and $r_i$ the distance between $p_i$ and $c(p_i)$. 
	By the definition of $c(p_i)$, the open disk centered at $p_i$ and with radius $r_i$ does not intersect $S(P)$. 
	This implies that the length of the boundary of $f_i$ is at least the perimeter of a disk with radius $r_i$, 
	that for both the $L_1$ and the $L_{\infty}$ metrics is $8r_i$. 
	That is, $|f_i| \le 8r_i$. 
	Moreover, the sum of the lengths of all the faces is $2|S(P)|$, 
	so $|f_1| + \cdots + |f_n| < 2|S(P)|$ since we do not sum the length of the unbounded face. 
	
	We define a multigraph $G$ whose vertex set is the union of
	the set of sites $P$, the set of vertices of $S(P)$, and $\{c(p_i)\in S \mid p_i \in P\}$. 
	The edge set of $G$ is the union of the set of edges of $S$ and 
	two parallel edges $\{p_i,c(p_i)\}$ for each $p_i \in P$. 
	The weight of each edge is its length in the drawing. 
	By the observations above, 
	$|G| = |S(P)| + 2r_1 + \cdots + 2r_n \leq |S(P)| + |f_1|/4 + \cdots + |f_n|/4 < |S(P)| + |S(P)|/2 = 1.5|S(P)|$. 
	
	To obtain the desired upper bound on $|\TSP (P)|$ it remains to show that $|\TSP (P)| \leq |G|$. 
	By construction, since $S(P)$ is Eulerian, so is $G$. 
	An Euler tour of $G$ defines a TSP tour for the vertices of $G$ by skipping vertices that were already visited (as in the Christofides 1.5-approximation algorithm for TSP on instances where the distances form a metric space~\cite{TSP_1.5}). 
	This TSP tour has length at most $|G|$ and can be shortcut so that it only visits the sites of $P$. 
	By the triangle inequality, the length of the tour does not increase with these shortcuts. 
	Thus, we have that $|\TSP (P)| \le |G| < 1.5|S(P)|$.
	
	The construction for the matching this bound is similar to the one in the proof of Theorem~\ref{thm:ratio}, but simpler. 
	An illustration can be found in Figure~\ref{fig:minsona_a}. 
	%Figure~\ref{fig:L1Linf}. 
	For every $\epsilon > 0$ we construct a set of sites  $P_\varepsilon$ (the set of TSP vertices/sona sites) such that $|\TSP (P_\varepsilon)| \ge (1.5 - \varepsilon) |S(P_\varepsilon)|$.
	 %
%	 \begin{figure}
%	 	\centering
%	 	\includegraphics[width=0.8\columnwidth]{L1Linf.pdf}
%	 	\caption{Construction matching the bound for the ratio between the lengths of the minimum-length TSP and the minimum-length sona drawing for a set of sites}
%	 	\label{fig:L1Linf}
%	 \end{figure}
	 %
	 
	We fix $k=\lceil 1/(2\varepsilon) \rceil$. 
	The set of sites $P_\varepsilon$ includes every integer lattice point $(x,y)$ 
	such that $|x|+|y| \leq k$. 
	Consider drawing $Q$ resulting from the union of the axis-aligned unit squares centered at these sites. 
	It is easy to see, for example by rotating the construction, 
	that so far we have added $(k+1)^2 + k^2$ sites to $P_\varepsilon$ 
	and that the length of $Q$ is $4(k+1)^2$. 
	Straightforward computations show that $\frac{4(k+1)^2 + (k+1)^2 + k^2}{4(k+1)^2} = 5/4 + \frac{k^2}{4(k+1)^2} \ge  5/4 + \frac{1}{4(2\varepsilon + 1)^2} = 1.5 - \varepsilon + \frac{\varepsilon^2 (4\varepsilon + 3)}{(2\varepsilon + 1)^2}  > 1.5 - \varepsilon$. 
	Thus, a dense-enough packing of sites along $Q$ yields the desired result. 	
\end{proof}

We next consider sona drawings on the sphere. 
By the definition of sona drawings in the plane, 
the unbounded face contains no sites. 
For the sphere we consider the following analogue: 
if there is a face that contains in its interior a half-sphere 
then this face contains no sites. 
Note that there is at most one such face. 
The following theorem shows a tight upper bound on the TSP/sona separation factor for drawings on the sphere. 
(We consider the usual metric inherited from the Euclidean metric in~$\mathbb{R}^3$.)

\begin{theorem}
	\label{thm:sphere}
	For drawings on the sphere, the length of the 
	minimum-length TSP tour for a set of sites $P$ 
	is at most $2$ times the length of the minimum-length sona drawing for $P$. 
	Moreover, this bound is tight. 
\end{theorem}

\begin{proof}
	The proof of the upper bound on the length of the  minimum-length TSP tour again follows the lines of the (unpublished) 
	proof of~\cite[Theorem 12]{DDD06}. 
	It only differs slightly % in one point 
	from the first part of the the proof of Theorem~\ref{thm:L1Linf}. 
	Using the same notation, in this case, 
	the distance $r_i$ between a site $p_i \in P$ 
	and its closest point $c(p_i)$ in $S(P)$ 
	corresponds to the length of the shortest arc on the great circle through $p_i$ and $c(p_i)$. 
	The open disk centered at $p_i$ and with radius $r_i$ is an open spherical cap that does not intersect $S(P)$. 
	Assuming that the sphere has radius $\rho$, 
	the boundary of this cap has length $2\pi \rho \sin (r_i / \rho)$. 
	Since the face containing a site cannot contain a half-sphere in its interior we have that 
	$0 \le r_i / \rho \le \pi/2$. 
	The function $\sin(x) / x$ in the interval $0\le x \le \pi/2$ is decreasing. 
	Thus, $\rho / r_i \sin (r_i / \rho) \geq 2/ \pi \sin (\pi/ 2) = 2/ \pi$. 
	This implies that $ 2\pi \rho \sin (r_i / \rho) \geq 4r_i$. 
	Thus, the face $f_i$ of $S(P)$ containing the site $p_i$ 
	has length $|f_i| \ge 4r_i$. 
	Moreover, $|f_1| + \cdots + |f_n| \le 2|S(P)|$. 
	With the same arguments and defining the same multigraph as in the proof of Theorem~\ref{thm:L1Linf} we obtain that 
	$|\TSP (P)| \le 2|S(P)|$.
	
	The construction showing that this bound is tight places two sites on the north and south poles of the sphere and packs the equator densely with sites. 
	Then the minimum-length sona drawing goes along the equator while the minimum-length TSP must reach both poles, yielding a $2-\varepsilon$ TSP/sona separation factor. 
\end{proof}

\section{Complexity of Length Minimization}
\label{sec:min-length}
\later{\section{Complexity of Length Minimization} \label{app:min-length}}

In this section and Appendix~\ref{app:min-length},
we prove that finding a sona drawing of minimum length for given sites is NP-hard, even when the sites lie on a polynomially sized grid. The complexity of minimum-length sona drawing was posed as an open problem in 2006 by Damian et al.~\cite[Open Problem 4]{DDD06}. We use a reduction from the problem of finding a Hamiltonian cycle in a grid graph (a graph whose $n$ vertices are a subset of the points in an integer grid, and whose edges are the unit-length line segments between pairs of vertices), proven NP-complete by Itai, Papadimitriou, and Szwarcfiter~\cite{IPS82}.

Let $V$ be the set of $n$ vertices in a hard instance for Hamiltonian cycle in grid graphs. If $V$ is a \textsc{yes} instance, its Hamiltonian cycle forms a Euclidean traveling salesman tour with length exactly $n$. If it is a \textsc{no} instance, the shortest Euclidean traveling salesman tour through its vertices has length at least $1$ for every grid edge, and length at least $\sqrt{2}$ for at least one edge that is not a grid edge (as this is the shortest distance between grid points that are non-adjacent), so its total length is at least $n+\sqrt{2}-1\approx n+0.414$. For the $L_1$ distance, the increase in length is larger, at least $1$.
Our reduction replaces each point of $V$ by two points, close enough together to make the increase in length from converting a TSP to a sona drawing negligible with respect to this gap in tour length.

\both{
\begin{theorem}
It is NP-hard to find a sona drawing for a given set of sites whose length is less than a given threshold $L$, for any of the $L_1$, $L_2$, and $L_\infty$ metrics.
\end{theorem}
}

\later{
\begin{proof}
Let $V$ be the set of $n$ vertices in a hard instance for finding a Hamiltonian cycle in grid graphs. 
We may form a hard instance of the minimum-length sona drawing problem for $L_1$ or $L_2$ distances  by replacing each vertex in $V$ by a pair of sites, one at the original vertex position and the other at distance less than $\varepsilon$ from it, where $\varepsilon=\Theta(1/n)$ is chosen to be small enough that $4n\varepsilon<\sqrt{2}-1$. We set $L=n+2n\varepsilon$. For $L_\infty$ distance, we use a hard instance for $L_1$ distance, rotated by $45^\circ$.

\begin{figure}[t]
\includegraphics[width=\columnwidth]{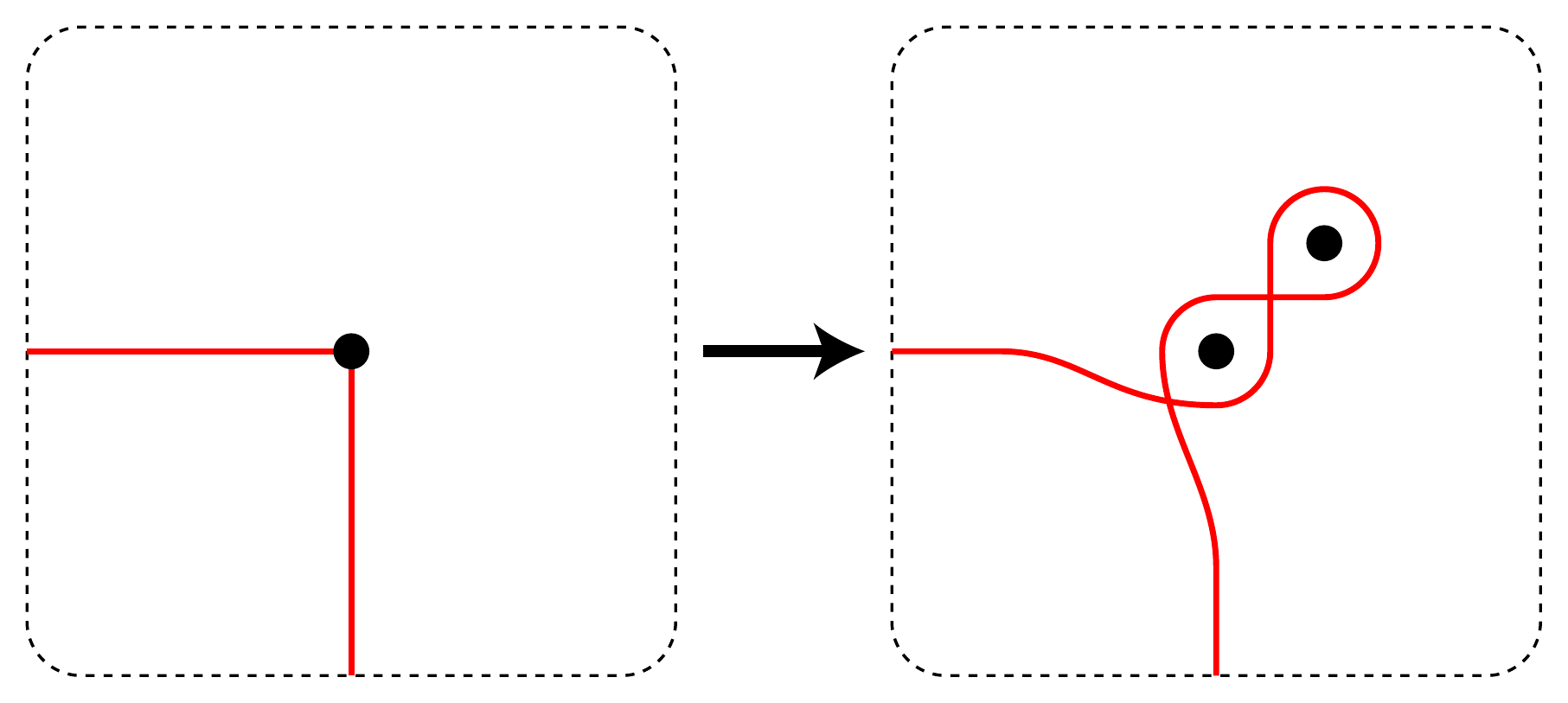}
\caption{Local modifications to convert a grid Hamiltonian cycle into a short sona drawing for a set of doubled sites}
\label{fig:ham2sona}
\end{figure}

If $V$ is a yes-instance for Hamiltonian cycle, let $C$ be a Hamiltonian cycle of length $n$ for $V$. We may form a sona drawing of length less than $L$ by modifying $C$ within a neighborhood of each pair of sites so that, for all but one of these pairs, it makes two loops, one surrounding each site (\autoref{fig:ham2sona}), and so that for the remaining pair it makes one loop around one of the two sites and surrounds the other point by the face formed by $C$ itself. In this way, each face of the modified curve surrounds a single site of our instance. Each of these local modifications to $C$ may be performed using additional length less than $2\varepsilon$, so the total length of the resulting sona drawing is less than $L$.

If $V$ is a \textsc{no} instance for Hamiltonian cycle, let $C$ be any sona drawing for the resulting instance of the minimum-length sona drawing problem. Then $C$ must pass between each pair of sites in the instance, and by making a local modification of length at most $2\varepsilon$ near each pair, we can cause it to touch the point in the pair that belongs to $V$ itself. Thus, we have a curve of length $|C|+2n\varepsilon$ touching all points of $V$. Because $V$ is a \textsc{no} instance, the length of this curve must be at least $n+\sqrt{2}-1$, from which it follows that the length of $C$ is at least $n+\sqrt{2}-1-2n\varepsilon\ge L$.
\end{proof}

By scaling the sites by a factor of $\Theta(1/\varepsilon)=O(n)$ we may obtain a hard instance of the minimum-length sona drawing problem in which all sites lie in an integer grid whose bounding box has side length $O(n^2)$.

It is possible to represent a minimum-length sona drawing combinatorially, as a \emph{conveyor belt}~\cite{BBD19} formed by bitangents and arcs of infinitesimally small disks centered at each site, and to verify in polynomial time that a representation of this form is a valid sona drawing. However, this does not suffice to prove that the decision version of the minimum-length sona drawing problem belongs to NP. The reason is that, when the sites have integer coordinates,
the limiting length of a sona drawing, represented combinatorially in this way, is a sum of square roots (distances between pairs of given points) and we do not know the computational complexity of testing inequalities involving sums of square roots~\cite{OR81,TOPP33}.
(Euclidean TSP has the same issue.)
}%\later

\section{Complexity of Grid Drawing Existence}
\label{sec:grid}

While minimizing the length of sona drawings in general is hard, if we restrict the drawing to lie on a grid, then even determining the existence of a sona drawing is hard. 

Given $n$ sites at the centers of some cells in the unit-square grid,
a \emph{grid sona drawing} is a sona drawing whose edges are drawn as
polygonal lines along the orthogonal grid lines (like orthogonal graph drawing).

We show that finding a grid sona drawing for a given set of sites is NP-hard by a reduction from Planar CNF SAT \cite{Lichtenstein-1982}.

\subsection{Construction}

%Here we describe the properties of the gadgets we will use for the reduction. 
In this section we view the grid as a graph, thus by \emph{edge} we mean a unit segment of a grid line, and by \emph{vertex} we mean a grid vertex --- these terms are distinct from ``sona edge'' etc. We say that an edge is either \emph{on} or \emph{off} according to as it belongs in the sona drawing. The subgraph of the grid that is on is the \emph{path graph}. Observe that two grid-adjacent sites always require the edge between them to be on; otherwise both would be in the same sona face (\emph{connected}). Also, every vertex must have even degree in the path graph.

Here we are concerned with internal properties of the gadgets. Their exteriors are lined with unconnected edges; we will show later how to connect them.

\begin{figure}[t]
\centering
\begin{minipage}[b]{.33\linewidth}
\centering
\includegraphics[scale=0.14]{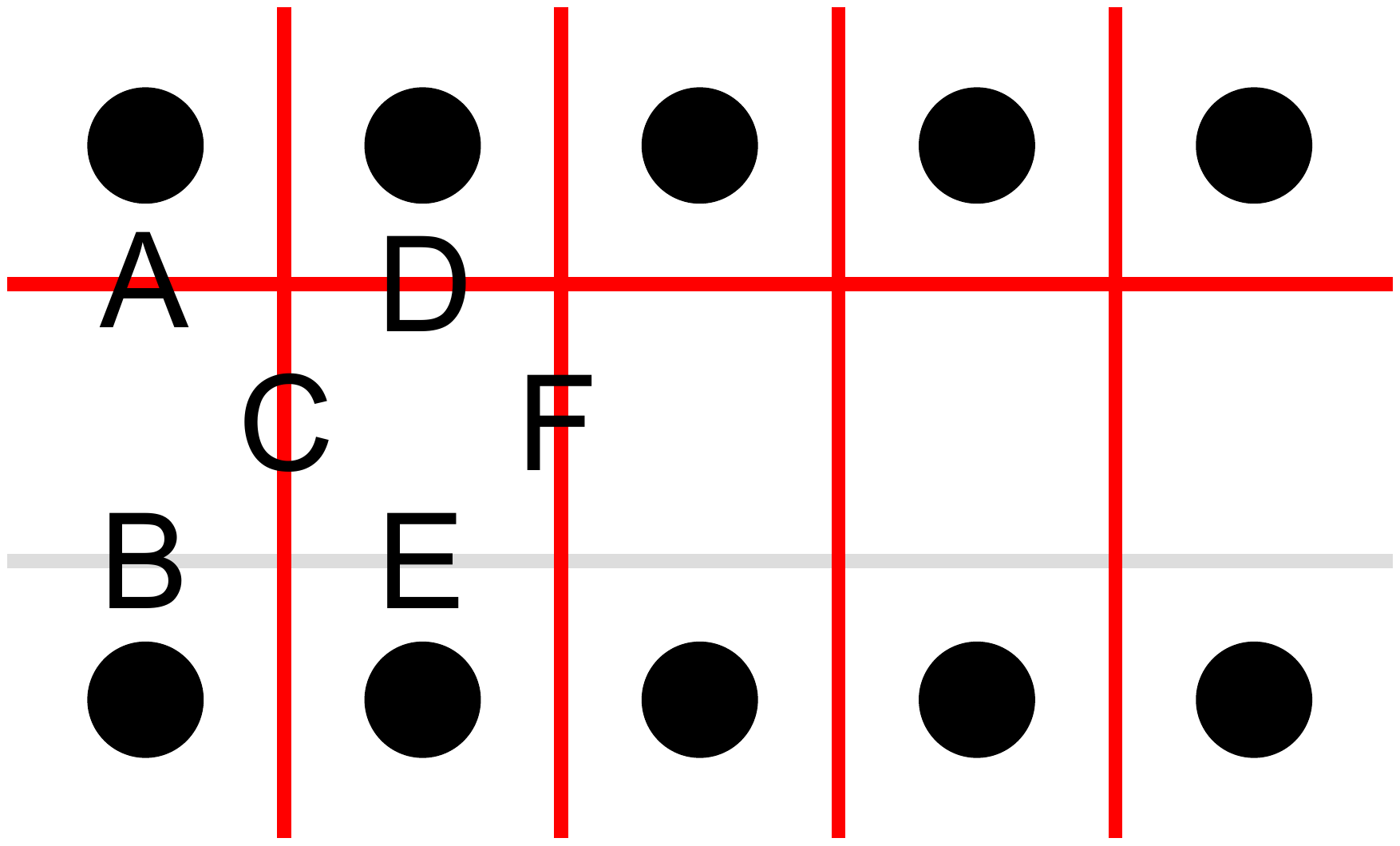}
\caption{\protect\raggedright Wire gadget\newline}
\label{fig:wire}
\end{minipage}\hfill
\begin{minipage}[b]{.33\linewidth}
\centering
\includegraphics[scale=0.14]{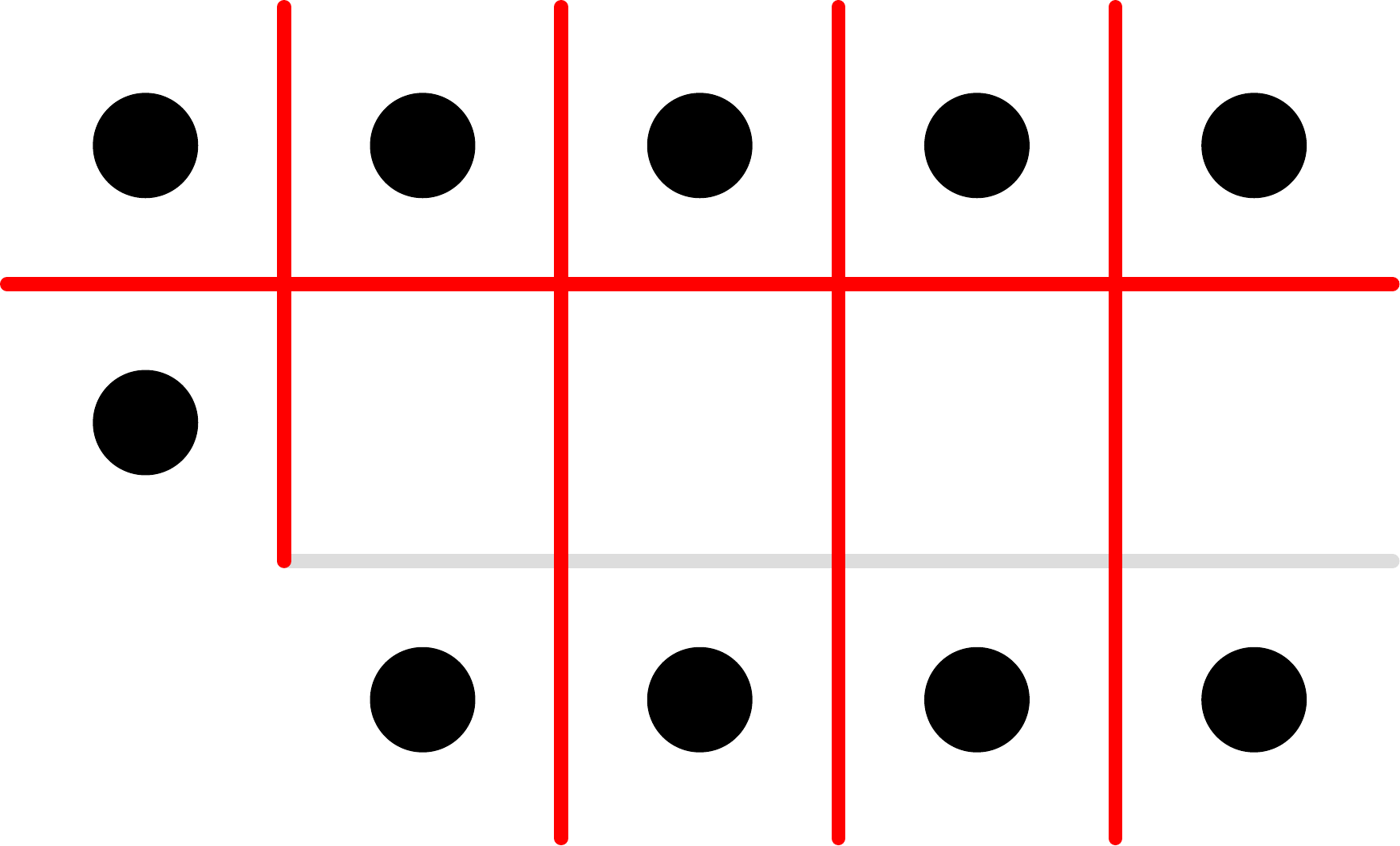}
\caption{\protect\raggedright Constant gadget\newline}
\label{fig:constant}
\end{minipage}\hfill
\begin{minipage}[b]{.3\linewidth}
\centering
\includegraphics[scale=0.11]{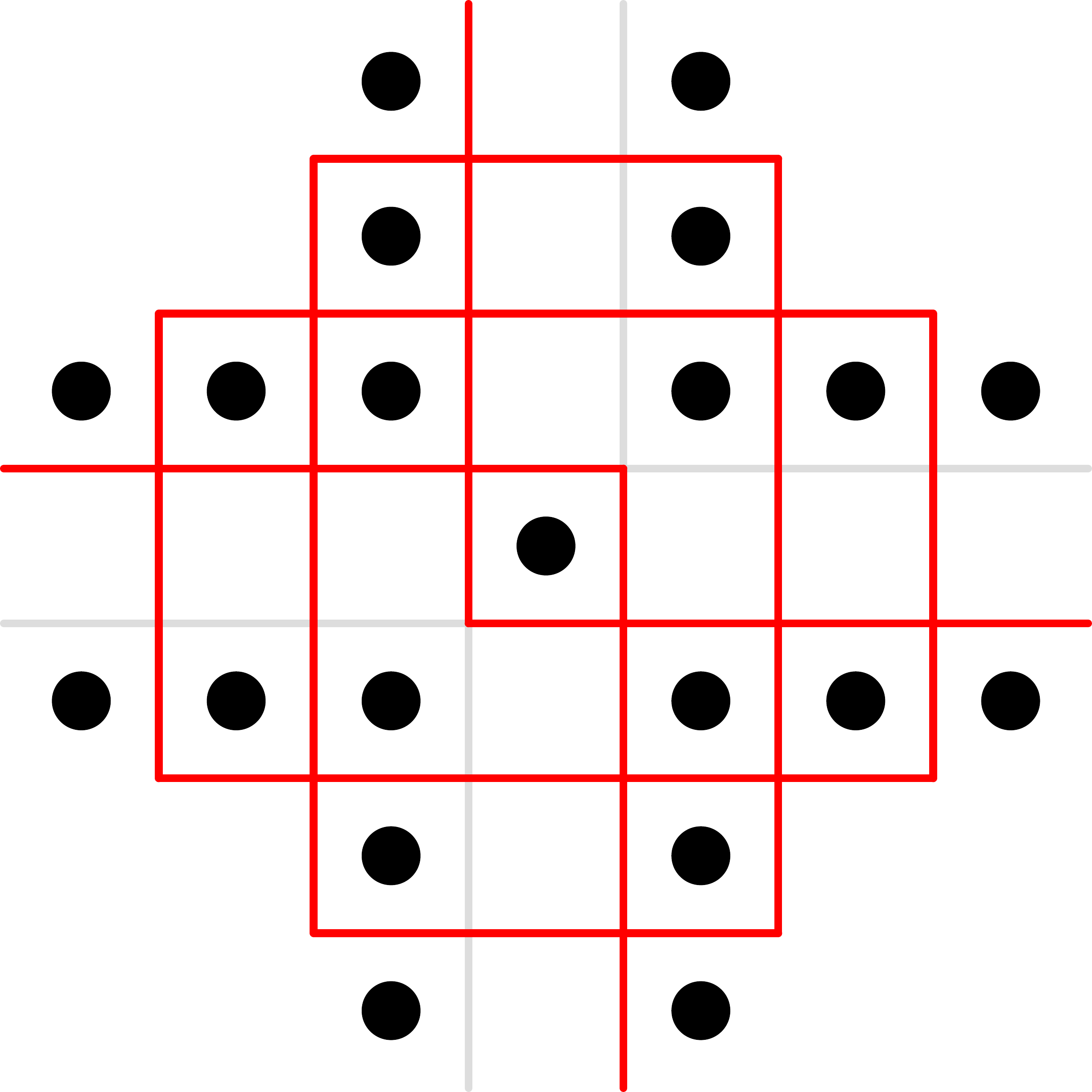}
\caption{\protect\raggedright Turn / Split / Invert gadget}
\label{fig:invert}
\end{minipage}
\end{figure}

%\begin{figure}[t]
%\centering
%\includegraphics[width=.4\columnwidth]{wire}
%\caption{Wire gadget}
%\label{fig:wire}
%\end{figure}

\paragraph*{Wire.}
The \emph{wire} gadget is shown in \autoref{fig:wire}. One of edges $A$ and $B$ must be on, otherwise two sites would be connected. Assume without loss of generality that $A$ is on. Now, suppose $C$ is off. Then $D$ must be too, to preserve even vertex degree. Then, $B$ and $E$ must both be on to prevent sites from being connected, but this is impossible with $C$ off. Therefore $C$ and $D$ are on. The same reasoning shows that the entire line $A$, $D$, etc. is on, as well as edges $C$, $F$, etc. Then $E$ must be off to prevent an empty face, and thus the entire line $B$, $E$, etc. is off. The wire thus has two states: the upper line can be on and the lower off, or vice-versa. We can extend a wire as long as necessary. An unconnected wire end serves as a variable.

In Figures~\ref{fig:wire} through~\ref{fig:or}, all marked edge states (red for on, gray for off) are forced by the indicated wire states. These marks were generated by computer search, but are easy to verify by local analysis.

%This is the type of local reasoning used to show all of our basic gadget properties. In the interest of space we will not be so explicit for most of the gadgets; simple inspection will confirm the claimed properties [ or do we want to put explicit proofs in the appendix? ]. All gadget properties were also confirmed with computer search of valid states.

%\begin{figure}[t]
%\centering
%\includegraphics[width=.4\columnwidth]{constant}
%\caption{Constant gadget}
%\label{fig:constant}
%\end{figure}

\paragraph*{Constant.}
The gadget shown in \autoref{fig:constant} forces the attached wire to be in the up state: the edge between the two left sites must be on, forcing the rest.

%\begin{figure}[t]
%\centering
%\includegraphics[width=.5\columnwidth]{turn}
%\caption{Turn Gadget}
%\label{fig:turn}
%\end{figure}

%\paragraph*{Turn.}
%The \emph{turn} gadget is shown in \autoref{fig:wire}. It is easily seen that each wire's state forces the other. For example, if the horizontal wire were in the lower state, a degree-3 vertex would be created at the right edge, forcing the vertical wire to be in the right state to rectify it.
%
%We use wires to represent truth values. When we turn, note that the representation reverses: if the signal comes in from the bottom, and the leftward state (relative to the direction of input) is ``true'', then as the signal exits to the left the rightward state is now ``true''.

%\begin{figure}[t]
%\centering
%\includegraphics[width=.5\columnwidth]{invert}
%\caption{Turn / Split / Invert gadget}
%\label{fig:invert}
%\end{figure}

\paragraph*{Turn / Split / Invert.}
The gadget shown in \autoref{fig:invert} is multi-purpose. If we view the left wire as the input, then the upper and lower outputs represent turned signals, and the right output represents an inverted signal. (Unused outputs can be left unattached, thus unconstrained.)

Any wire state forces all the others. Given that the left wire is in the up state, suppose the right wire is also up. Then the top wire and bottom wire must be in the same left/right state, otherwise we will have degree-three vertices in the middle. But this would leave the central site connected to another site, so the right wire is forced down. Then, if the (top, bottom) wires are not in the (left, right) states, again the central site will be connected to another one. (This figure contains multiple loops, but these will be eliminated in the final configuration by adding more edges.)

\begin{figure*}[t]
\centering
\subcaptionbox{\label{fig:or-1} false + true $\rightarrow$ true}
  {\includegraphics[scale=0.105]{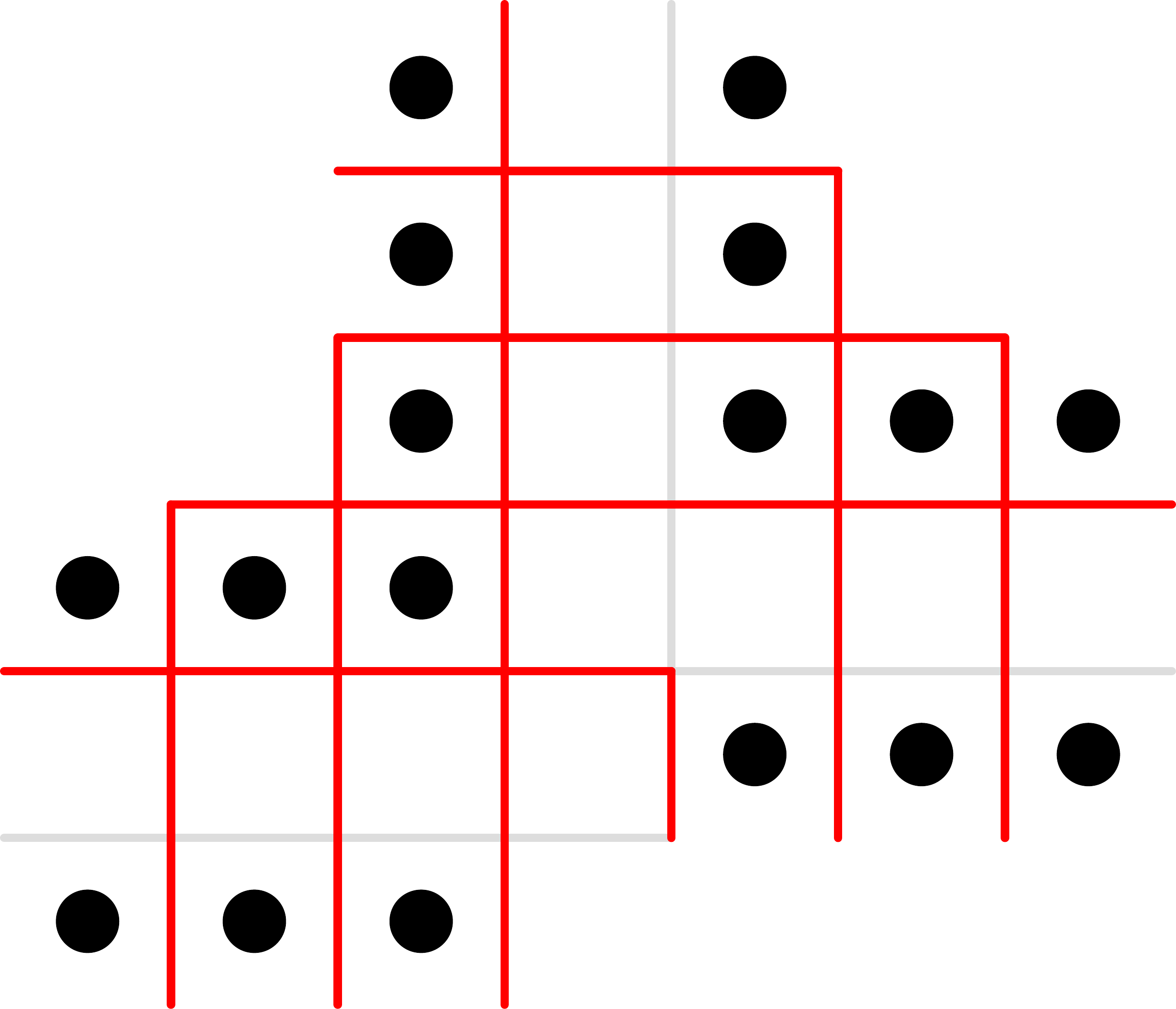}}\hfill
\subcaptionbox{\label{fig:or-2} true + true $\rightarrow$ true}
  {\includegraphics[scale=0.105]{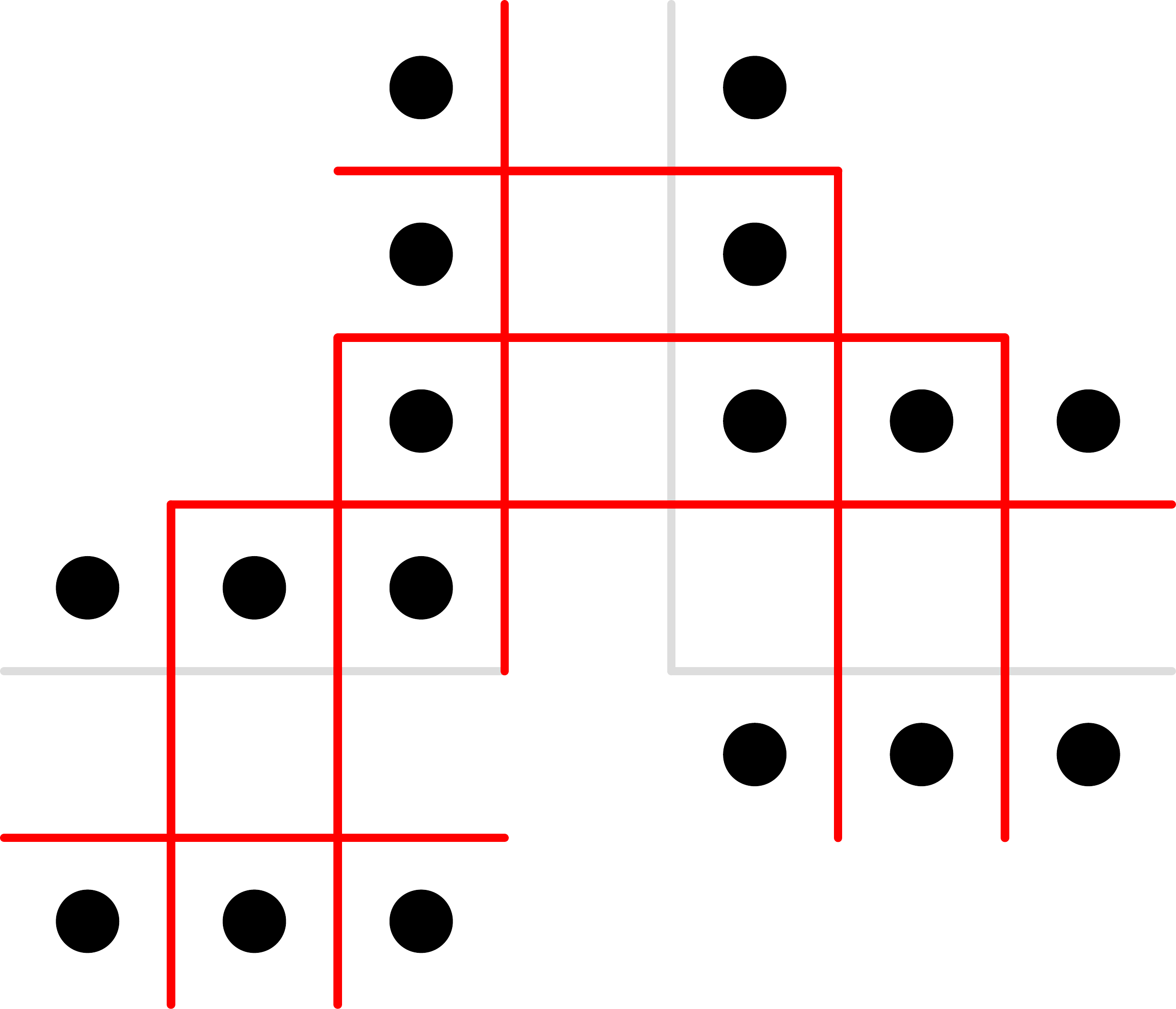}}\hfill
\subcaptionbox{\label{fig:or-3} true + false $\rightarrow$ true}
  {\includegraphics[scale=0.105]{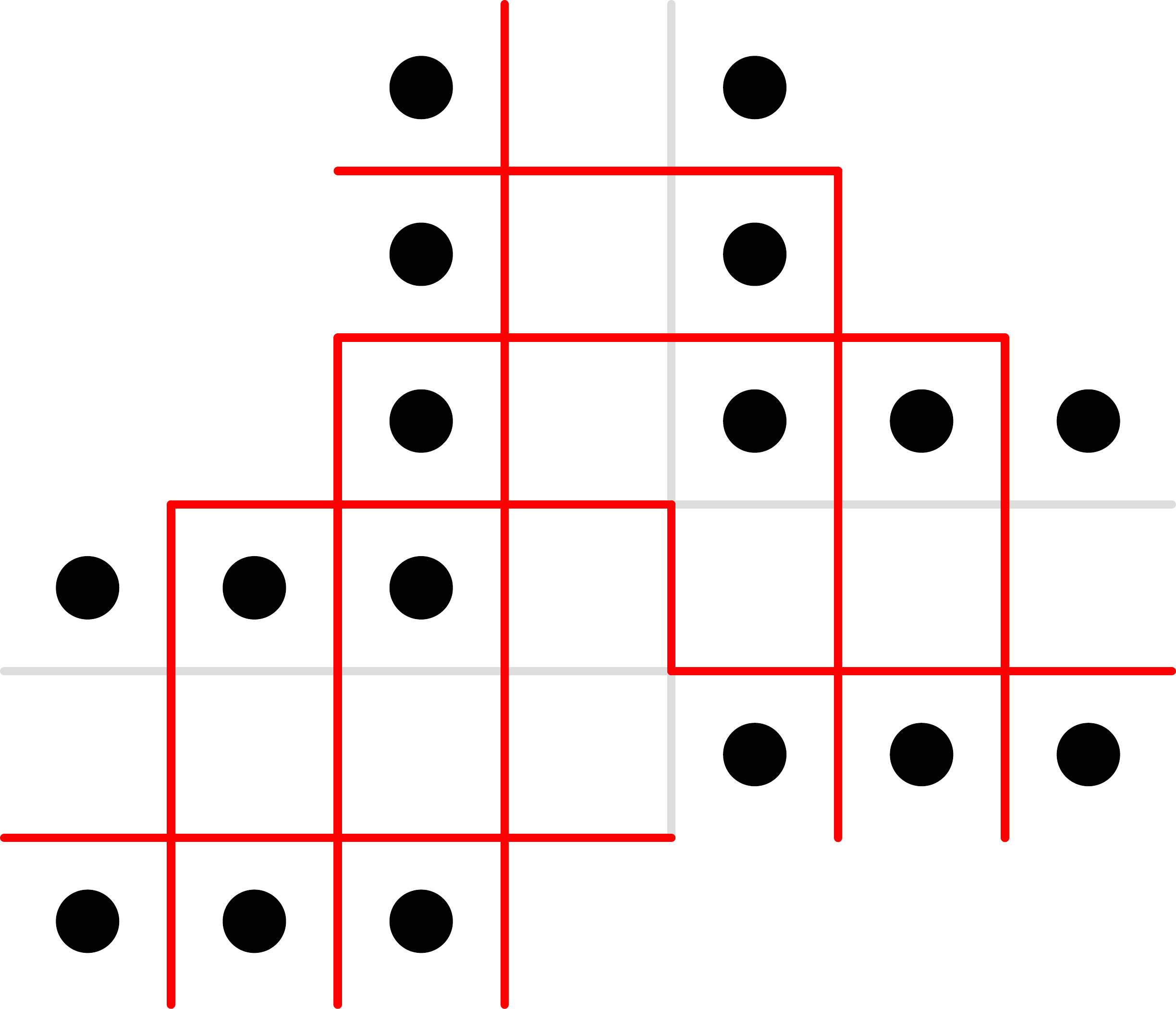}}\hfill
\subcaptionbox{\label{fig:or-4} false + false $\rightarrow$ false}
  {\includegraphics[scale=0.105]{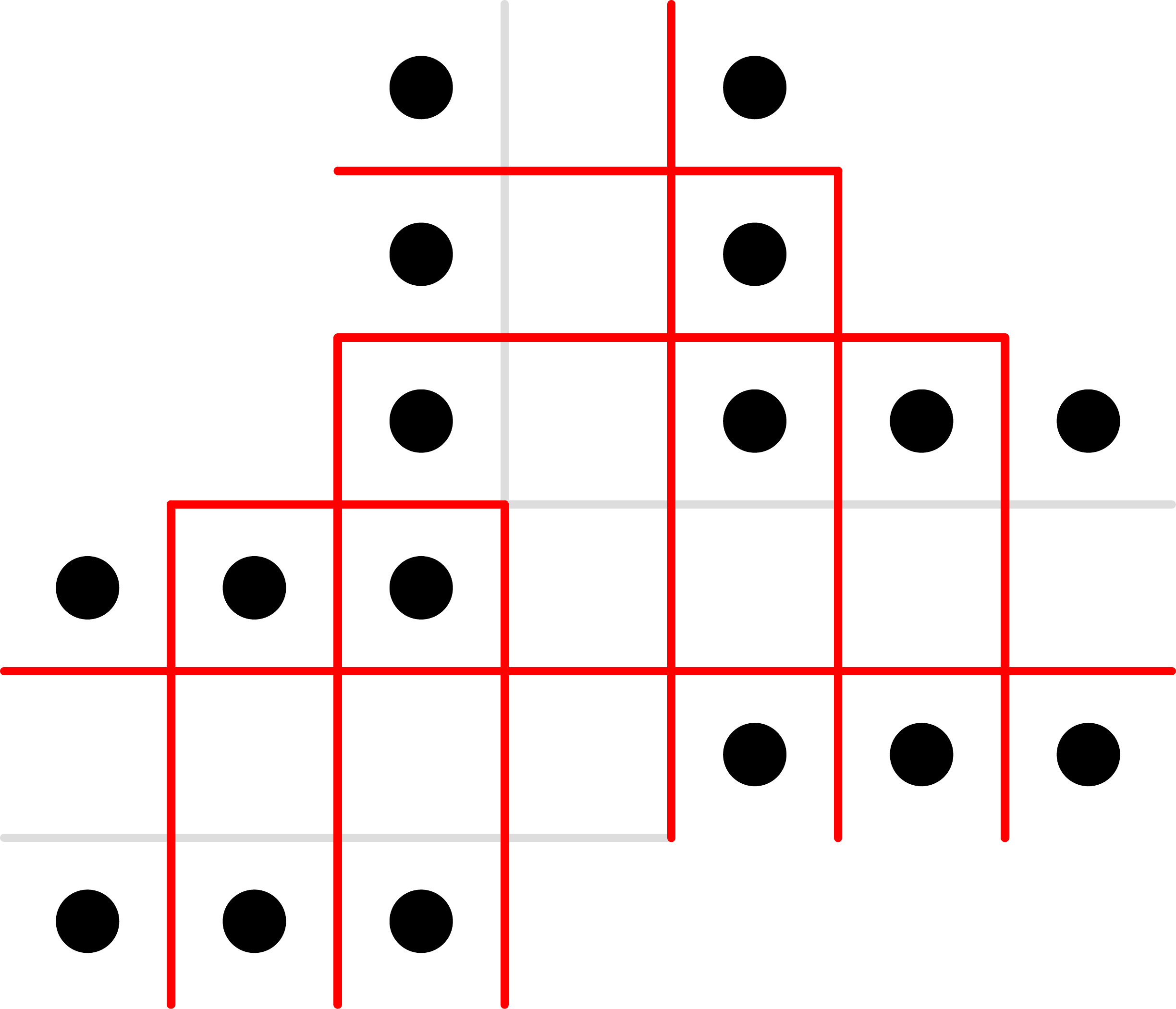}}\hfill
\subcaptionbox{\label{fig:or-bad} Bad state}
  {\includegraphics[scale=0.105]{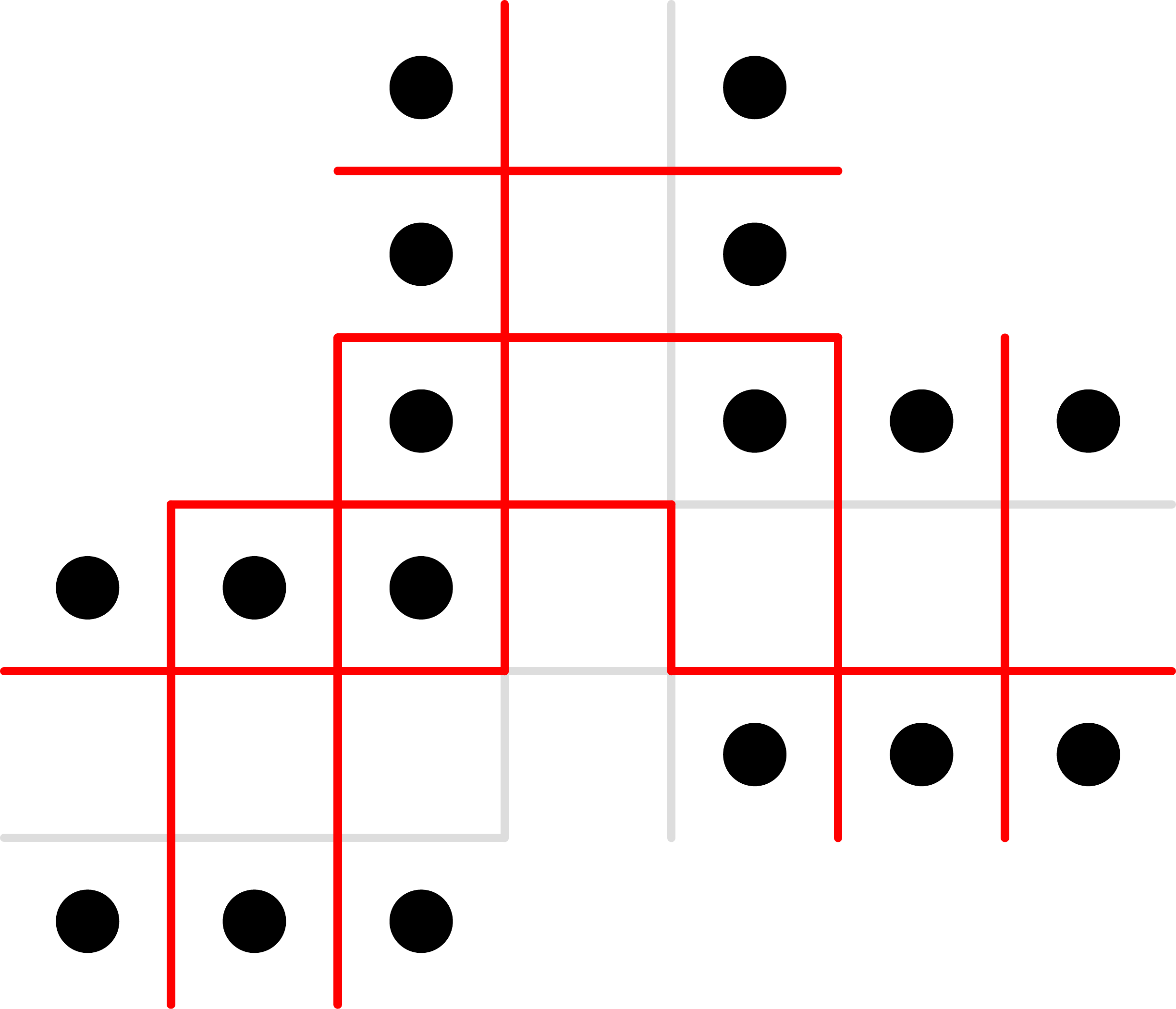}}
\caption{OR gadget}
\label{fig:or}
\end{figure*}

%\begin{figure}[t]
%\centering
%\includegraphics[width=.4\columnwidth]{or-small}
%\caption{OR gadget, shorter wires. Is this sufficient? We probably still need to add ellipses at the wire ends regardless.}
%\label{fig:or-small}
%\end{figure}

\paragraph*{OR.}
\autoref{fig:or} shows the OR gadget. The upper wire is interpreted as an output, with the left state representing true; the other wires are inputs, with true represented as down on the left wire and up on the right wire. (We can easily adjust truth representations between gadgets with inverters.) If either input is true, then the output may be set true, as shown. If both inputs are false, the output may be set false. \autoref{fig:or-bad} shows that setting the output to true when both inputs are false is not possible: all marked edge states are forced by the wire properties, but two sites are left connected.

\subsection{Hardness}

\begin{theorem}
It is NP-hard to find a grid sona drawing for a given set of sites at the centers of grid cells.
\end{theorem}

\begin{proof}
%We wire the gadgets together.
Given a CNF Boolean formula with a planar incidence graph, we connect the above gadgets to represent this graph: unconstrained wire ends represent variables, and are connected to splitters and inverters to reach clause constructions. A clause is implemented with chained OR gadgets, with the final output constrained to be true with a constant gadget. By the gadget properties described above, we will be able to consistently choose wire states if and only if the formula is satisfiable.

We must still show that all edges can be joined together into a single closed loop, while retaining the sona properties. Our basic strategy for connecting loose ends is to border each gadget with ``crenellations'', as shown in \autoref{fig:crenellations}. This figure also shows how to pass pairs of path segments across a wire without affecting its internal properties, which we will use to help form a single loop.

\begin{figure}[t]
\centering
\includegraphics[scale=0.105]{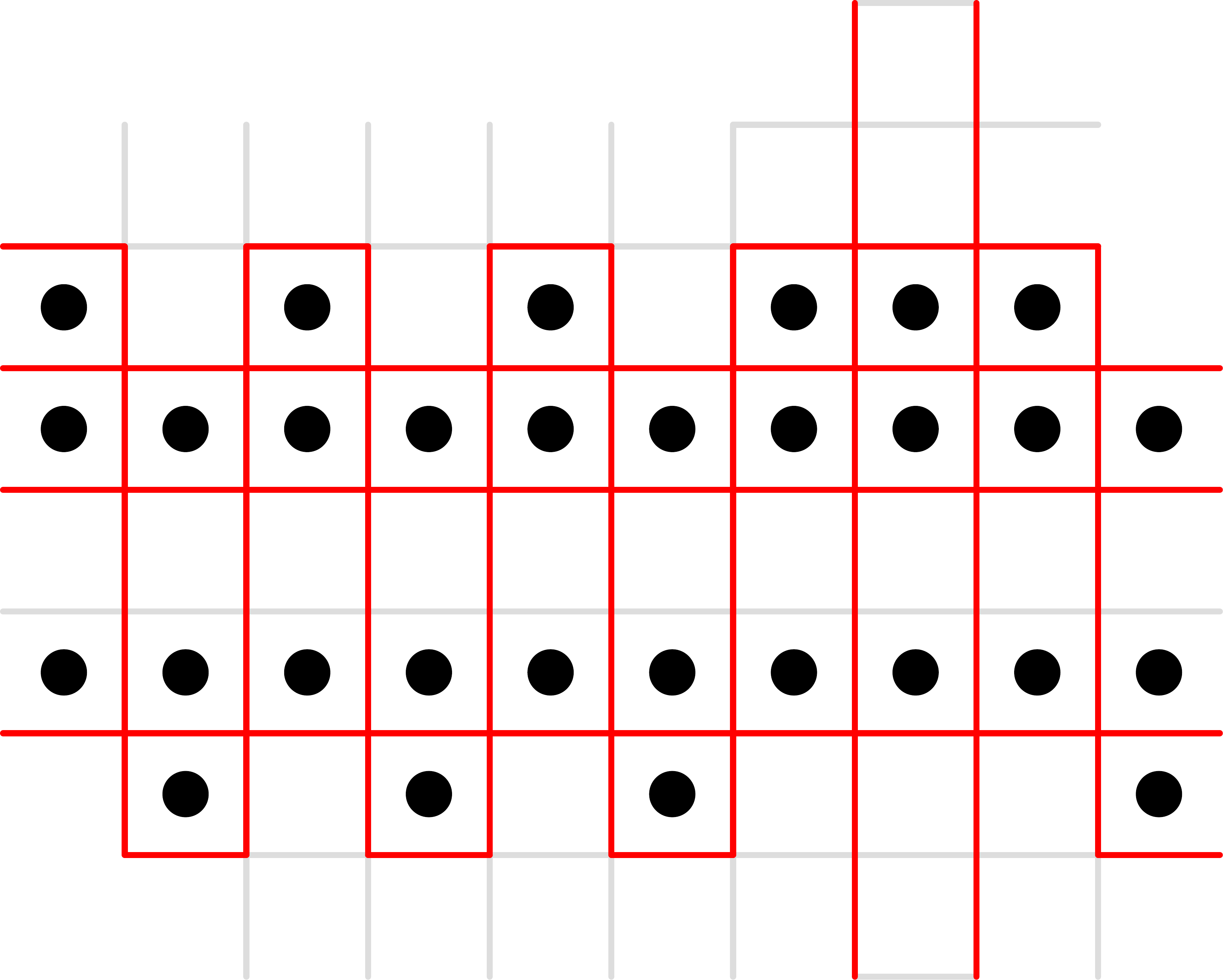}
\caption{Wire gadget with crenellations and pass-through}
\label{fig:crenellations}
\end{figure}

%\begin{figure}[t]
%\centering
%\includegraphics[width=.7\columnwidth]{loose-ends}
%\caption{Ending a wire with unattached edges}
%\label{fig:loose-ends}
%\end{figure}

\begin{figure}[t]
\centering
\subcaptionbox{\label{fig:turn-1} State 1}
  {\includegraphics[scale=0.105]{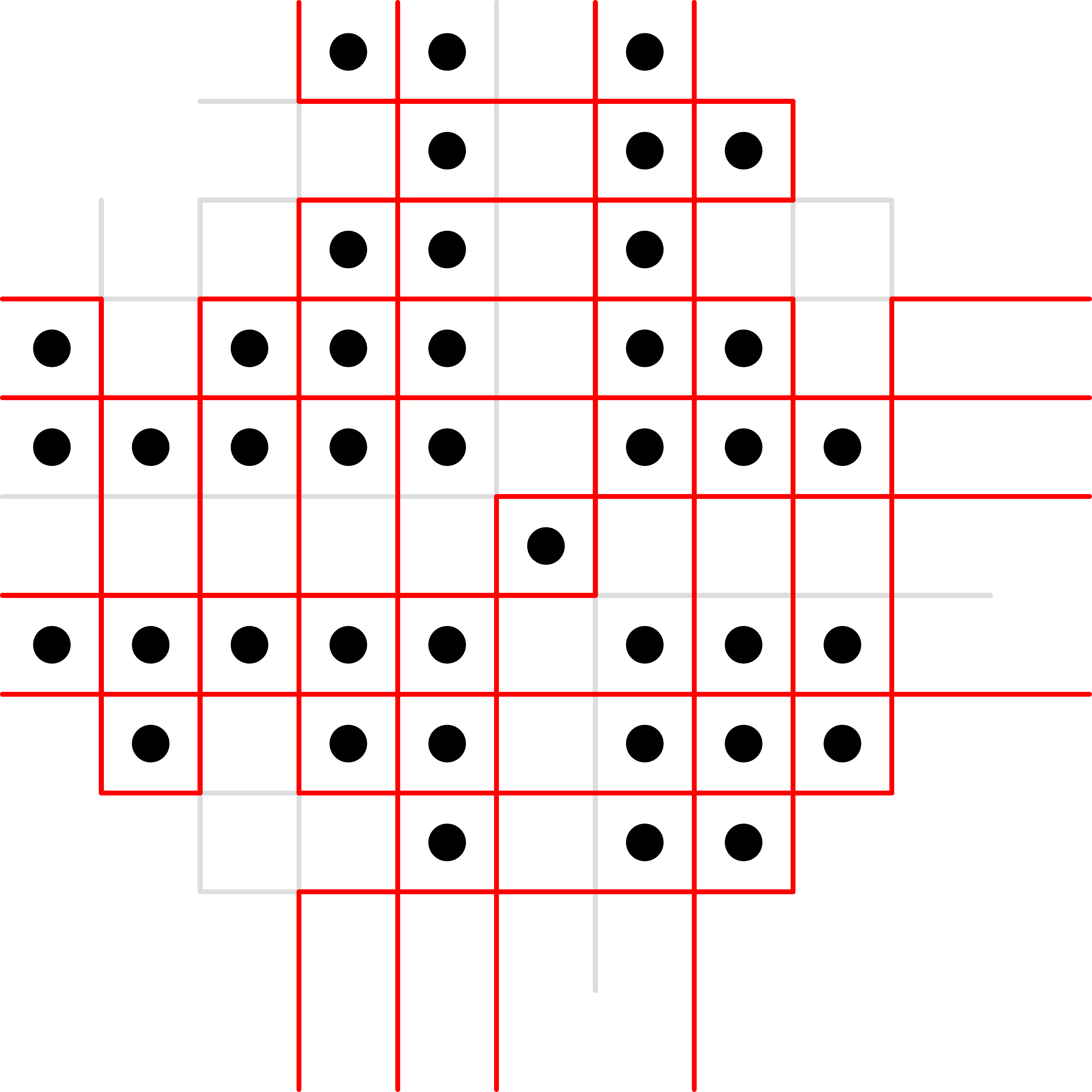}}\hfill
\subcaptionbox{\label{fig:turn-2} State 2}
  {\includegraphics[scale=0.105]{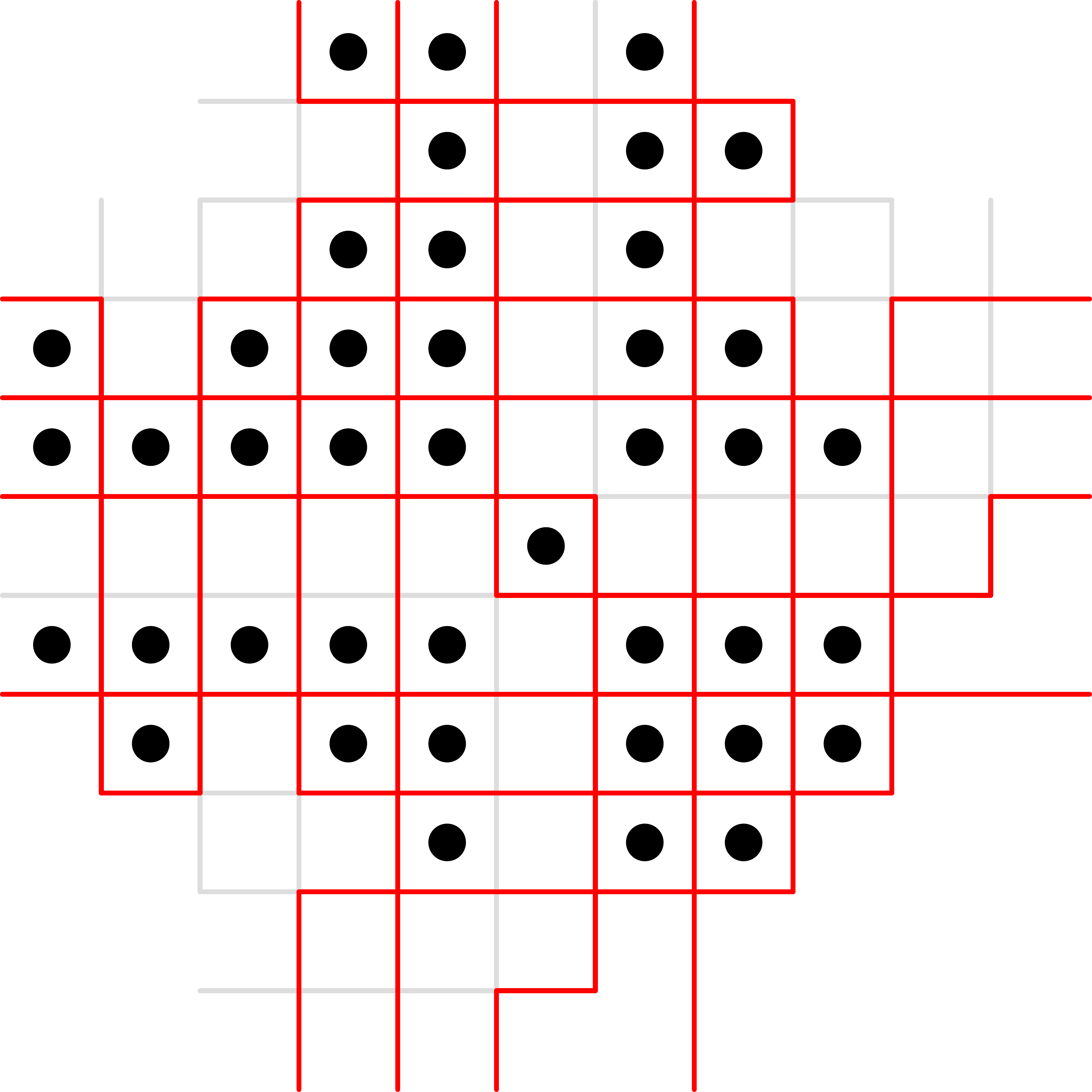}}
\caption{Crenellated turn}
\label{fig:turn}
\end{figure}

% Now, when we connect all the gadgets together, there are no loose ends. However, we may still have more than one loop. We solve this by re-introducing loose ends at every unconstrained wire end, as shown in \autoref{fig:loose-ends}, and then re-connecting them to form a single loop, as described below. 

%At every turn, there are two loose wire ends from the unused ports in \autoref{fig:invert}. 

Adding crenellations to the other gadgets is straightforward, and we defer explicit figures to Appendix~\ref{app:crenellations}, with one exception. (The crenellations do add a parity constraint when wiring gadgets together; we show in the appendix how to shift parity.) When we use the gadget in \autoref{fig:invert} to turn a wire, it will be useful to use the crenellated version in \autoref{fig:turn}. With the connections to other gadgets on the left and top, the right and bottom portions are unconstrained. We can place edges as shown, so that they leave the gadget identically regardless of which state it is in. Then, all paths entering from the left or the top leave on the bottom or the right as loose edges, except that in \autoref{fig:turn-1}, one path connects the left to the top. If we connect a right turn to the top port, this path will also terminate in an unconnected edge. If every wire contains a left turn and matching right turn, then every path in the sona graph must end in two unconnected edges in turn gadgets, because there are no internal loops in any of the gadgets. 

The space occupied by the loose ends of a turn lies either in an internal face of the wiring graph, or on its exterior. We route the interior ends to pass-through pairs as shown in \autoref{fig:crenellations}, so all unconnected edges wind up on the outer border of the graph. Because the terminal edges are placed identically in Figures~\ref{fig:turn-1} and \ref{fig:turn-2}, we can plan their routing without knowing the wire states. As a result, we can place additional sites as required for the property that a single site lies in each internal sona face. (We can lengthen the wires as needed to create additional routing space in the internal faces.)

Now we are in a state where all paths end on the exterior of the construction. If we join these paths together without crossing, the number of extra sites needed in the outer face is just the number of paths. We place that many sites in a widely spaced grid (spacing proportional to number of paths) surrounding the inner construction. Then, we can complete the path greedily by repeatedly connecting one outer path end to one of its neighboring path ends, surrounding one of the added sites. Only one of its two neighboring path ends can come from the same path, so there's always another one to connect to. The wide grid spacing of the outer sites means there is always room to route the connection.
%
% use the following procedure to connect the ends up into a single loop. Pick an unterminated edge. The unterminated edge either to its left or to its right 
%
%Importantly, the internal routing plan, and thus the location of the additional sites in the internal faces and the external space, are known as a function of the formula layout, and do not depend
%
%
\end{proof}

\section*{Acknowledgments}

Thanks to Godfried Toussaint for introducing us (and computational geometry)
to sona drawings.
This research was initiated during the Virtual Workshop on Computational
Geometry held March 20--27, 2020,
which would have been the 35th Bellairs Winter Workshop on
Computational Geometry co-organized by E. Demaine and G. Toussaint
if not for other circumstances.
We thank the other participants of that workshop for helpful discussions
and providing an inspiring atmosphere.

{\raggedright
\bibliographystyle{plainurl}
\bibliography{sona}
}

\appendix
\magicappendix

\section{Crenellations for Grid Drawing}
\label{app:crenellations}

Figures~\ref{fig:crenellated-constant},  \ref{fig:crenellated-variable}, and \ref{fig:crenellated-or} show how to add crenellations to the Constant gadget, an unconstrained wire end (variable), and the OR gadget, respectively. The crenellated Split / Invert is the same as in \autoref{fig:turn}, extended in the obvious way for ports that are used. In no case do the crenellations affect the internal properties described in the main text; these figures simply show that it is possible to add the crenellations appropriately.

As mentioned in the main text, the crenellations do add a parity constraint when connecting gadgets with wires; we can no longer make wires of arbitrary length, but must match the crenellations to the gadgets at each end. In order to do that we need one additional gadget, an inverting turn, shown in \autoref{fig:inverting-turn}. Unlike in \autoref{fig:turn}, the wire state is switched during the turn. Observe that in \autoref{fig:turn}, turning does not change crenellation parity, but the straight-through path, which would invert if not terminated, does change crenellation parity. The inverting turn also does not change crenellation parity. Therefore, to change the crenellation parity of a wire, we can invert it (straight through), changing the parity, and add a sequence left inverting turn, right turn, right turn, left turn to restore the original line of the wire.

\begin{figure}[p]
\centering
\includegraphics[scale=0.13]{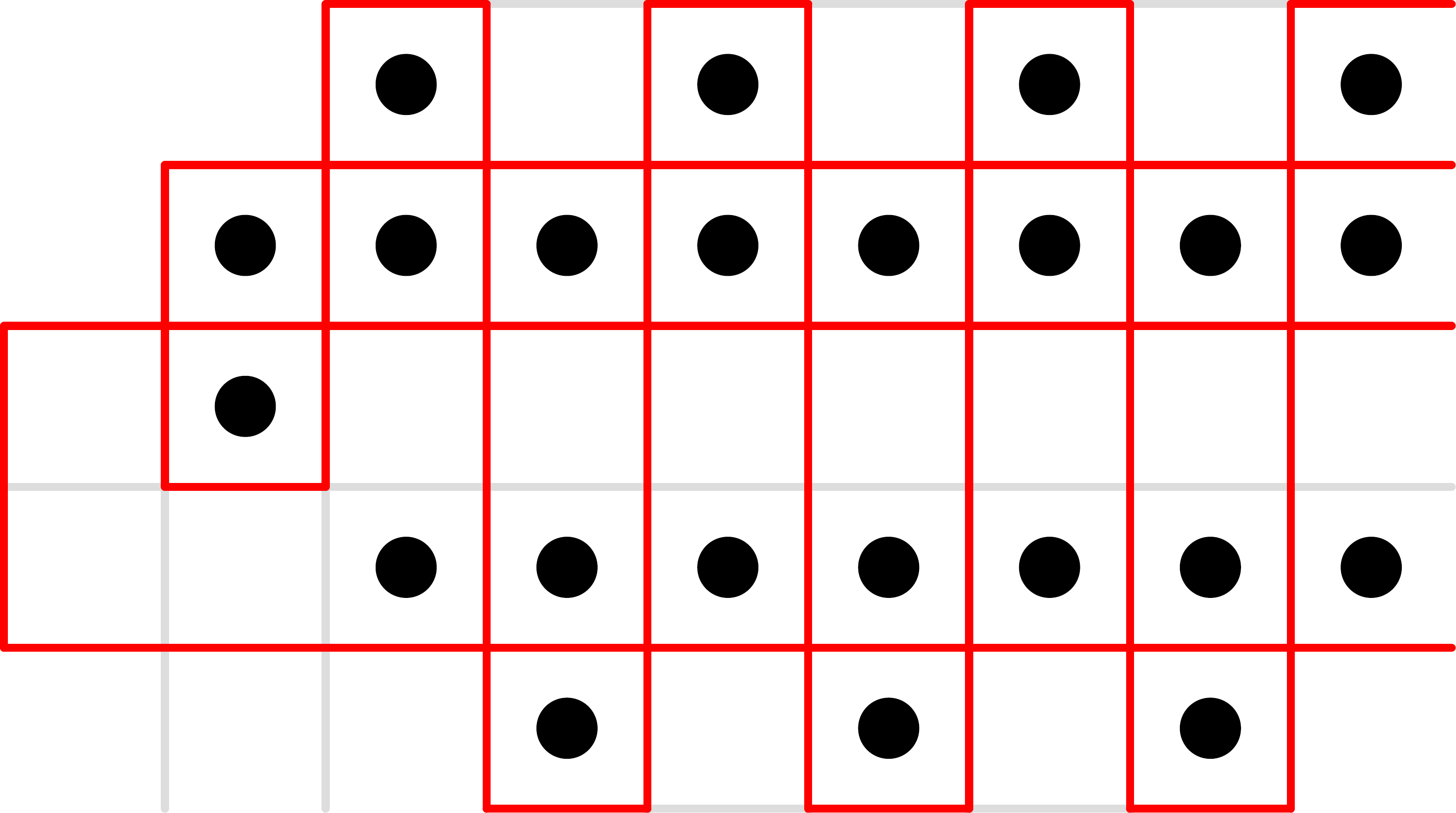}
\caption{Crenellated Constant gadget}
\label{fig:crenellated-constant}
\end{figure}

\begin{figure}[t]
\centering
\subcaptionbox{State 1}
  {\includegraphics[scale=0.13]{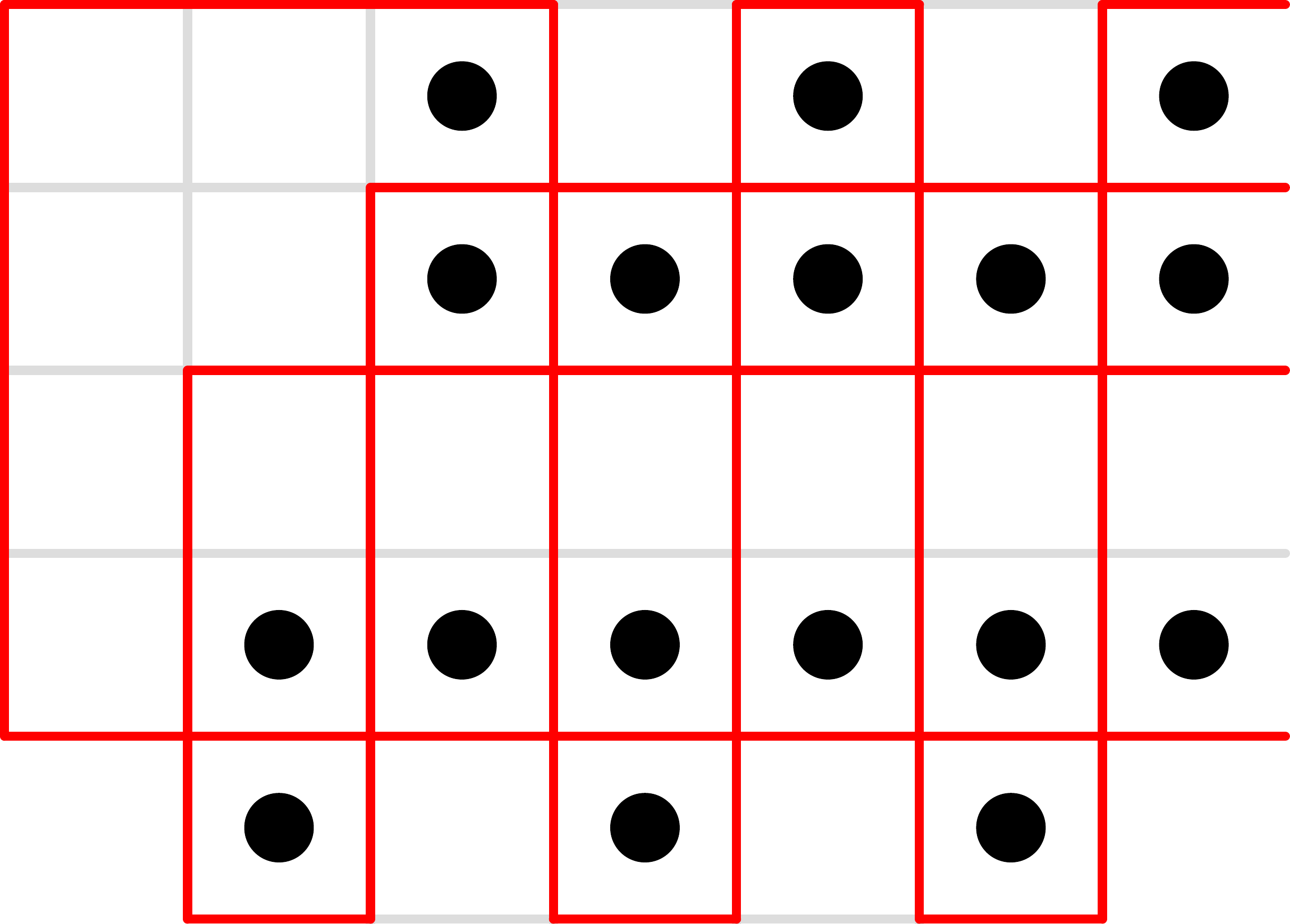}}\hfil\hfil
\subcaptionbox{State 2}
  {\includegraphics[scale=0.13]{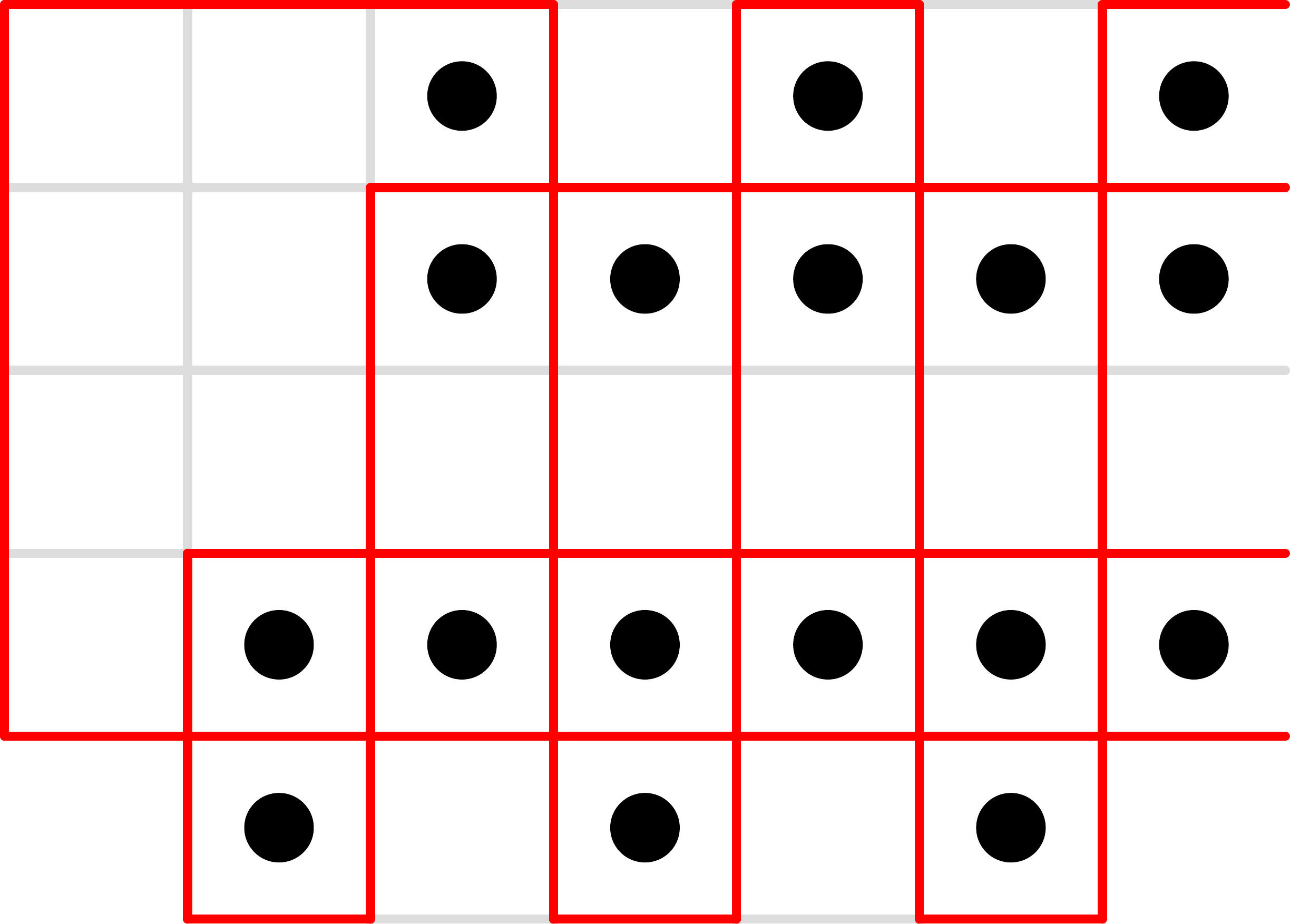}}
\caption{Crenellated unconstrained wire end}
\label{fig:crenellated-variable}
\end{figure}

\begin{figure}[t]
\centering
\subcaptionbox{\label{fig:crenellated-or-1} false + true $\rightarrow$ true}
  {\includegraphics[scale=0.13]{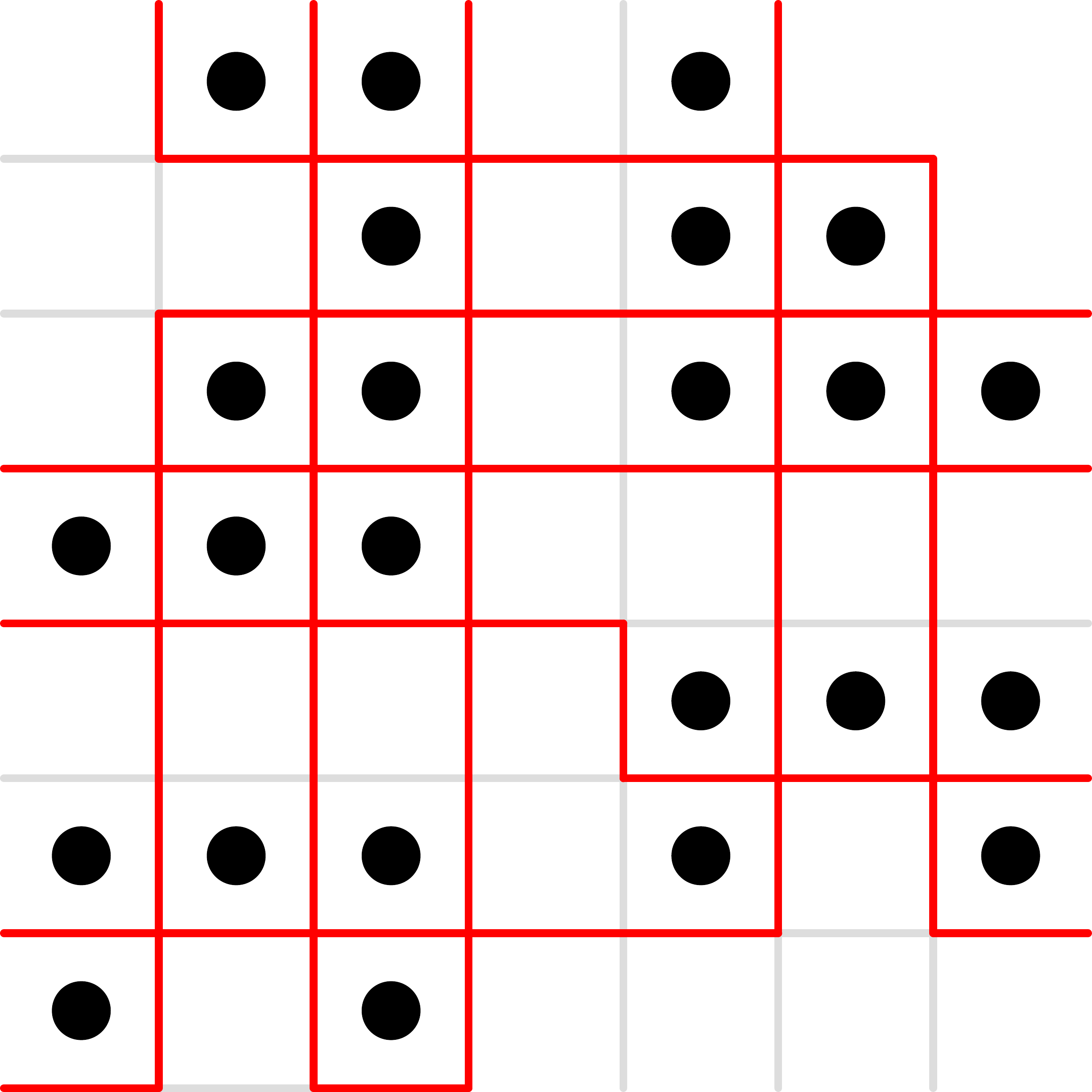}}\hfil\hfil
\subcaptionbox{\label{fig:crenellated-or-2} true + true $\rightarrow$ true}
  {\includegraphics[scale=0.13]{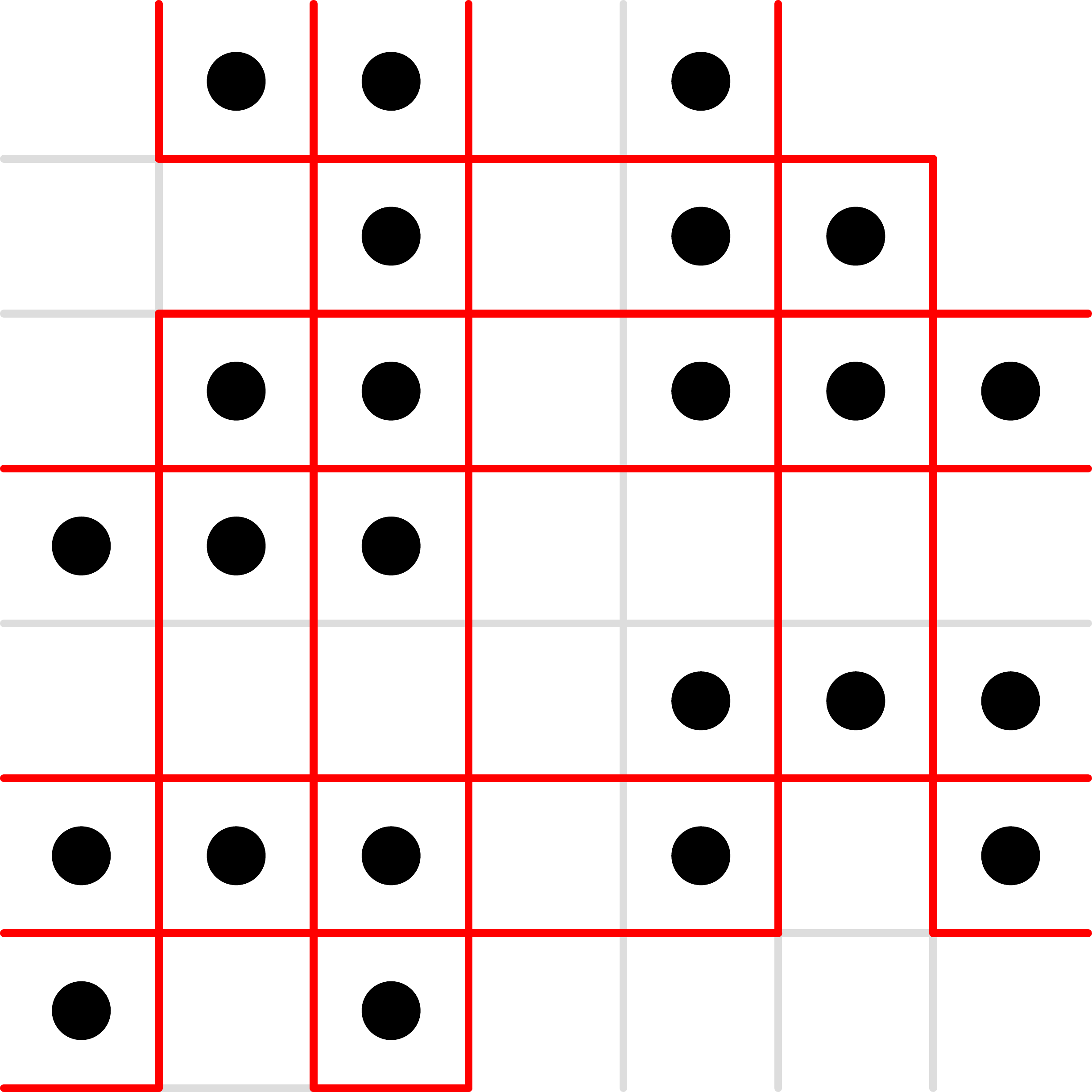}}

\medskip

\subcaptionbox{\label{fig:crenellated-or-3} true + false $\rightarrow$ true}
  {\includegraphics[scale=0.13]{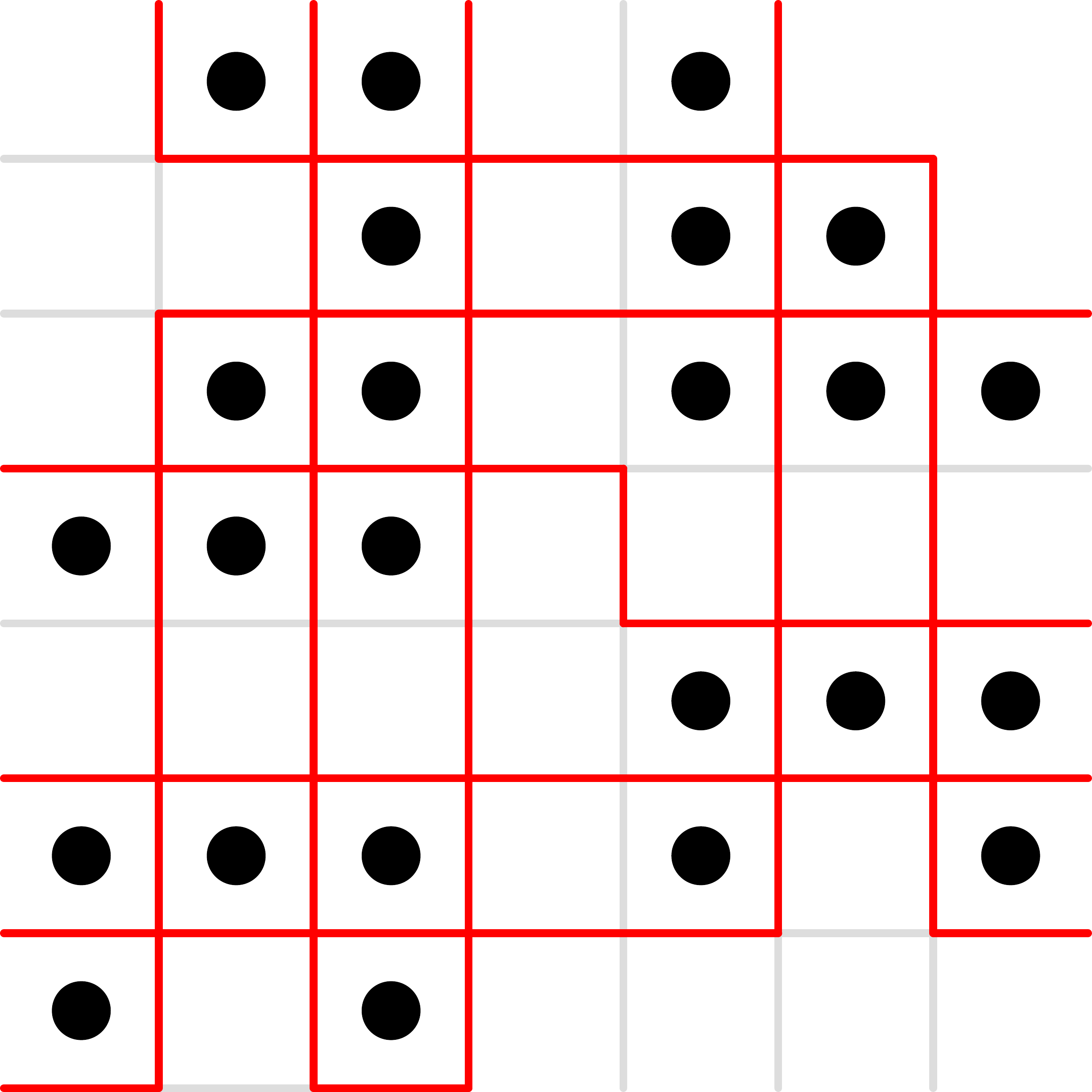}}\hfil\hfil
\subcaptionbox{\label{fig:crenellated-or-4} false + false $\rightarrow$ false}
  {~\includegraphics[scale=0.13]{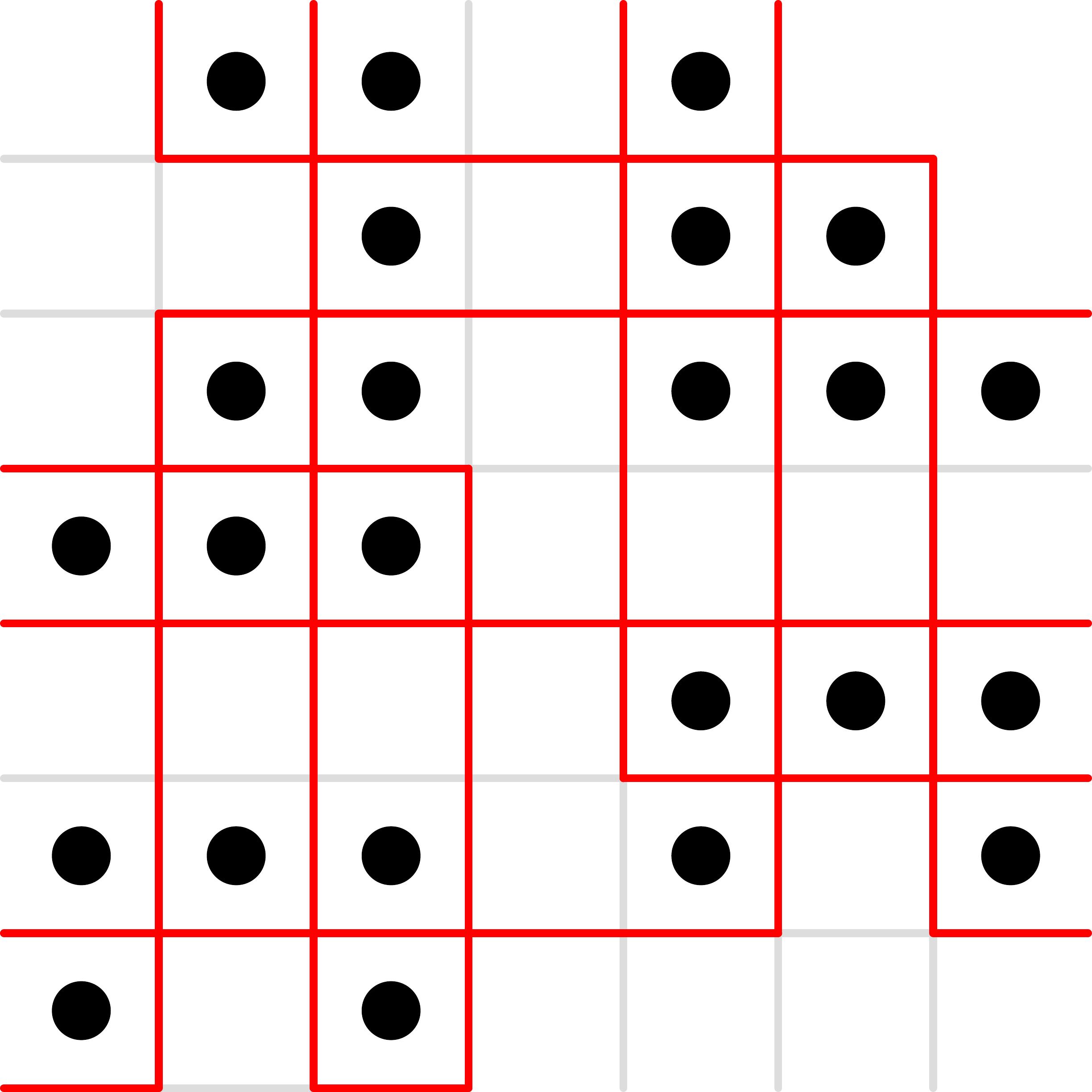}~}
\caption{Crenellated OR gadget}
\label{fig:crenellated-or}
\end{figure}

\begin{figure}[t]
\centering
\includegraphics[scale=0.13]{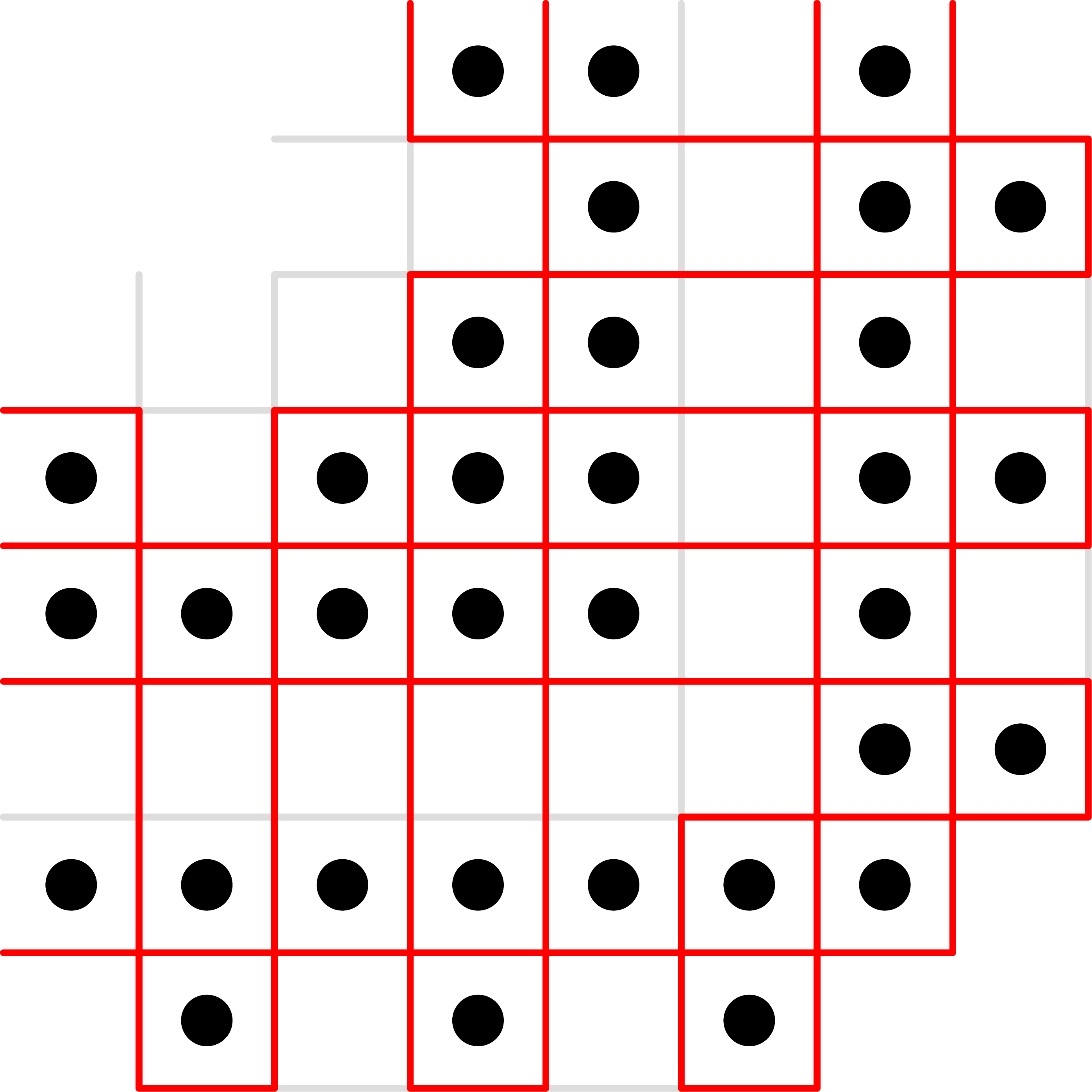}
\caption{Crenellated inverting Turn gadget}
\label{fig:inverting-turn}
\end{figure}

\end{document}